\documentclass[11pt]{article}
\usepackage[T1]{fontenc}
\usepackage{lmodern}
\usepackage[margin=1in]{geometry}
\usepackage{amsmath,amssymb,amsthm,mathtools}
\usepackage{booktabs,array,longtable}
\usepackage{graphicx,xcolor,tikz}
\usetikzlibrary{arrows.meta,positioning,calc,fit,backgrounds}
\usepackage{microtype}
\usepackage{needspace}
\usepackage{enumitem}
\usepackage[colorlinks=true,linkcolor=blue!45!black,citecolor=blue!45!black,urlcolor=blue!45!black]{hyperref}
\usepackage{url}
\newtheorem{theorem}{Theorem}
\newtheorem{lemma}[theorem]{Lemma}

\newtheorem{corollary}[theorem]{Corollary}
\theoremstyle{definition}
\newtheorem{definition}[theorem]{Definition}

\DeclareMathOperator{\conv}{conv}

\DeclareMathOperator{\lk}{lk}

\setlist{itemsep=3pt,topsep=5pt}
\title{Flip Graphs for Eight Points in Three Dimensions Are Connected}
\author{Marc Khoury\\[3pt]\small\href{mailto:khoury@eecs.berkeley.edu}{\nolinkurl{khoury@eecs.berkeley.edu}}}
\date{}
\begin{document}
\maketitle
\begin{abstract}
We prove that every configuration of four through eight points in
three-dimensional space, with no four points coplanar, has a connected
full geometric flip graph under $2 \leftrightarrow 3$ flips. Every point
remains fixed and present throughout the sequence. The proof brings each
tetrahedralization into placing form at a convex hull vertex: the
tetrahedra not incident to that vertex fill the convex hull of the
remaining points. This reduction allows us to establish connectivity by
induction, using the connectivity of regular tetrahedralizations. The
main geometric tool is radial projection, which turns the tetrahedra
incident to a hull vertex into a planar triangulation. When this planar
triangulation is regular, varying its lifting heights produces legal
spatial flips that progressively shrink the region occupied by the
incident tetrahedra and reach placing form. Planar lifting criteria and
compatibility constraints between the projections at different hull
vertices resolve the remaining small cases. For eight points, at most 35
flips are needed to reach placing form.
\end{abstract}
\section{Introduction}
A geometric flip replaces a tetrahedralization of a convex five-point
bipyramid by its other tetrahedralization. Although each replacement is
local, connectivity requires control of the surrounding mesh: a
combinatorial replacement need not have the correct geometric support.
We prove the following statement with every vertex retained.

\begin{theorem}[Complete eight-point connectivity]\label{thm:main}
For every finite $P\subset\mathbb R^3$ with $4\leq |P|\leq8$ and no four
distinct points coplanar, the graph $\mathcal F(P)$ of full geometric
tetrahedralizations using exactly $P$, with exact supported geometric
$2\leftrightarrow3$ circuit flips as edges, is connected.
\end{theorem}
There is no restriction on the numbers of hull and interior points.
Cospherical subsets are allowed. All paths keep the coordinates fixed;
they use no Steiner points, $1\leftrightarrow4$ moves, or nonsimplicial
spatial states.

The proof turns a suitable hull vertex into a \emph{placing} vertex:
its antistar fills the convex hull of the other points, and its remaining
tetrahedra are cones over visible boundary facets. Induction supplies a
flip path that makes the tetrahedralization of this smaller convex hull
regular while keeping the tetrahedra incident with the placing vertex
fixed. All original labels remain present. The resulting full
tetrahedralization is also regular, by the placing-extension argument in
Lemma~\ref{lem:placing}. Since all full regular tetrahedralizations are
connected by flips (Lemma~\ref{lem:regular}), every starting
tetrahedralization reaches the same flip component, even if different
starts use different placing vertices.

The principal local result, Theorem~\ref{thm:radial}, applies at any
cardinality. If the actual radial link of a hull vertex with $m$ mesh
neighbors is strictly regular, at most $\binom m4$ legal flips reach
placing form. Each move strictly shrinks the geometric star and adds no
neighbor. The endpoint has a global convex unit-pivot height. Forest
lifting and a short planar ear argument supply such links through seven
points and for every eight-point hull distribution with at least five
hull vertices. The tetrahedral-hull case requires two further ideas.
With at most three interior neighbors, positive reciprocal targets allow
one preparatory exchange. With complete hull/interior adjacency, the
four radial links cannot all have nonregular small-disk forms.

The proof also gives at most $35$ flips to placing form for eight points
(Corollary~\ref{cor:35}). This is neither a bound to a regular
triangulation nor a diameter bound.

\subsection{Related work and conventions}
Azaola and Santos prove connectivity for vector configurations of corank
three and $3$-connectivity in the acyclic case
\cite[Theorem 2 and Corollary 3.10]{AS2000}. For a $d$-dimensional point
configuration this concerns $d+4$ labels, hence seven in dimension three.
Their convention allows triangulations whose vertex set is a subset of
the configuration. We include a direct seven-point proof that explicitly retains every
label and use it in the induction for eight points, but do not claim
novelty for the seven-point result.

Regular triangulations and their secondary fans provide the standard
lifting framework \cite{DRS2010}. Pournin and Liebling prove that two
regular triangulations can be connected while preserving their common
faces, and deduce connectivity with a prescribed common vertex set
\cite[Theorem 8 and Corollary 9]{PL2007}. We give the needed full-label argument directly in
Lemma~\ref{lem:regular}. Santos~\cite{Santos2006} explains the distinction between geometric and abstract flips. Few available geometric moves do not
imply disconnection: De Loera, Santos, and Urrutia construct
three-dimensional triangulations with flip deficiency, including
configurations with one interior point \cite{DSU1999}.

The forest criterion used here is proved through positive equilibrium
weights and integration of affine lifting data. This is part of the
planar lifting framework discussed in \cite[Chapters 3 and 5]{DRS2010};
the criterion is not presented as a new planar regularity conclusion.
The contribution developed in this article is the self-contained route
from radial liftings and small-link compatibility to the stated
fixed-label connectivity theorem. We make no claim that the resulting flip bound is optimal.

\section{Geometric flips, regular triangulations, and placing induction}
\label{sec:prelim}
Throughout, $P\subset\mathbb R^3$ has $n\geq4$ points and no four distinct
points are coplanar. A \emph{full geometric tetrahedralization} $T$ is a
simplicial complex with $V(T)=P$, underlying space $|T|=\conv(P)$,
nondegenerate tetrahedral maximal simplices, and common-face intersections.
We write a simplex by juxtaposing its vertices. Put
\[
 N_T(v)=\{u:uv\in T\},\quad m_T(v)=|N_T(v)|,\quad
 \tau_T(v)=|\{\sigma\in T:\dim\sigma=3,\ v\in\sigma\}|.
\]
Thus $m_T(v)$ is a mesh-neighbor count and $\tau_T(v)$ is a tetrahedral
degree. For a hull vertex $q$, let $b(q)$ be its degree in the graph of
$\partial\conv(P)$. The degree of $T$ as a vertex of the flip graph counts
legal flips and is a fourth, different quantity. The link is
$\lk_T(v)=\{\sigma:v\notin\sigma,\ \sigma\cup\{v\}\in T\}$.
The geometric star $S_T(v)$ is the union of tetrahedra containing $v$;
the antistar is the subcomplex of faces not containing $v$.

\begin{lemma}[Exact circuit replacement]\label{lem:circuit}
Let $Q\subset P$ have five elements and let $Q_+\mid Q_-$ be its Radon
partition. Its two triangulations are
\[
 \mathcal T_+(Q)=\{Q\setminus\{a\}:a\in Q_+\},\qquad
 \mathcal T_-(Q)=\{Q\setminus\{a\}:a\in Q_-\}.
\]
If $|Q_+|=2$, $|Q_-|=3$, and one complete side consists of tetrahedra of
$T$, replacement by the other side is a legal full geometric flip.
\end{lemma}
\begin{proof}
By general position, $Q$ is a full-dimensional circuit. Santos
\cite[Lemma~1.13 and Definition~1.14]{Santos2006} gives its two
triangulations and the replacement rule, provided the maximal simplices
on the old side have a common link in $T$. Here those simplices are
tetrahedra and hence maximal in $T$; each has link $\{\varnothing\}$,
so the condition is automatic. The rule therefore produces a geometric
triangulation by replacing precisely the displayed old side and leaving
the exterior unchanged. For a $2\mid3$ partition, each side uses all five
labels, so the resulting triangulation is still full. A $1\mid4$ move
instead activates or deactivates its interior label and is excluded
from $\mathcal F(P)$.
\end{proof}
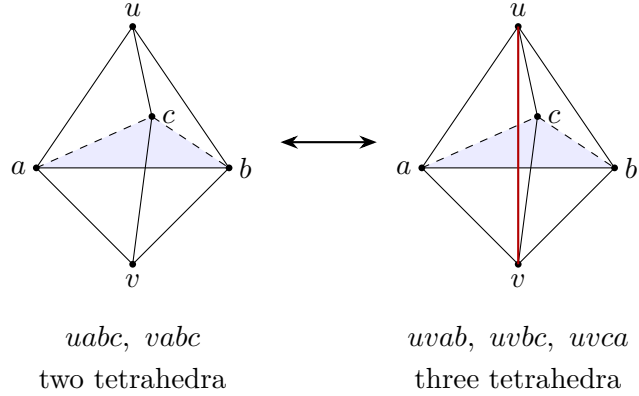
\begin{figure}[htbp]\centering
\begin{tikzpicture}[scale=.85,line join=round]
\foreach \dx in {0,6}{
\begin{scope}[xshift=\dx cm]
\coordinate (a) at (-1.5,0);\coordinate (b) at (1.5,0);
\coordinate (c) at (.3,.8);\coordinate (u) at (0,2.2);\coordinate (v) at (0,-1.5);
\fill[blue!8] (a)--(b)--(c)--cycle;
\draw (a)--(b)--(u)--(a)--(v)--(b);
\draw (u)--(c)--(v);\draw[dashed] (a)--(c)--(b);
\foreach \p in {a,b,c,u,v}\fill (\p) circle (1.5pt);
\node[left] at (a) {$a$};\node[right] at (b) {$b$};\node[right] at (c) {$c$};
\node[above] at (u) {$u$};\node[below] at (v) {$v$};
\end{scope}}
\draw[<->,>=Stealth,thick] (2.3,.4)--(3.8,.4);
\draw[red!70!black,thick] (6,2.2)--(6,-1.5);
\node[align=center,anchor=north] at (0,-2.3) {$uabc,\ vabc$\\[2pt]two tetrahedra};
\node[align=center,anchor=north] at (6,-2.3) {$uvab,\ uvbc,\ uvca$\\[2pt]three tetrahedra};
\end{tikzpicture}
\caption{Schematic $2\leftrightarrow3$ circuit replacement. The Radon
partition is $\{u,v\}\mid\{a,b,c\}$; legality requires the complete old
side in the mesh. The drawing does not certify that partition.}
\label{fig:flip}
\end{figure}

A full triangulation is \emph{strictly regular} if some heights $h_p$
have exactly its tetrahedra as projected lower faces of
$\conv\{(p,h_p):p\in P\}$. Equivalently, their piecewise-affine
interpolant is convex with strict folds across every interior triangular
face. Upper regularity uses a concave interpolant and is equivalent by
negating heights. A weak convex interpolant may have flat folds and does
not establish strict regularity.

\begin{lemma}[The full regular locus]\label{lem:regular}
Full strictly regular tetrahedralizations exist, and they form one
connected subgraph of $\mathcal F(P)$.
\end{lemma}
\begin{proof}
For $n=4$ the sole tetrahedron gives the conclusion, so suppose $n\geq5$.
We first describe the heights that keep every label on the lower hull.
For each nonextreme point $p$, the set
\[
 \Lambda_p=\left\{(\alpha_u)_{u\ne p}:\alpha_u\geq0,\quad
 \sum_{u\ne p}\alpha_u=1,\quad
 \sum_{u\ne p}\alpha_u u=p\right\}
\]
is a nonempty compact polytope of convex representations by the other
points. For any such representation, the other lifted points have a
convex combination at spatial position $p$ and height
$\sum_{u\ne p}\alpha_u h_u$. Its minimum over $\Lambda_p$ is therefore
the lower envelope at $p$ formed without $p$. The lift $(p,h_p)$ is an
exposed lower vertex precisely when it lies strictly below this envelope:
\begin{equation}\label{eq:activity}
 h_p<\sum_{u\ne p}\alpha_u h_u\qquad(\alpha\in\Lambda_p).
\end{equation}
Strict inequality permits a separating affine plane through the lift
and below all other lifted points. Equality with the minimum makes the
lift a convex combination of other lifts; a height above that minimum
cannot lie on the lower envelope. Hull labels impose no restrictions:
choose an affine function $\ell$ with $\ell(p)=0$ and $\ell(u)<0$ for
$u\ne p$. For sufficiently large $M$, the plane of height
$h_p+M\ell(x)$ supports the lifted hull from below only at $(p,h_p)$.

A linear function on $\Lambda_p$ attains its minimum at a vertex, so
\eqref{eq:activity} need only be tested at its finitely many vertices.
The set $\mathcal A\subset\mathbb R^P$ of heights keeping every label
an exposed lower vertex is consequently a finite intersection of strict
linear halfspaces. It is open and convex. It is also nonempty: for
$h_p=\|p\|^2$ and every $\alpha\in\Lambda_p$,
\[
 \sum_{u\ne p}\alpha_u\|u\|^2-\|p\|^2
 =\sum_{u\ne p}\alpha_u\|u-p\|^2>0.
\]
The last inequality uses $u\ne p$, nonnegative weights, and their sum
being one. Membership in $\mathcal A$ preserves vertices but may still
allow nonsimplicial lower faces.

To obtain triangulations, for each five-point set $Q\subset P$ fix its
affine dependence $d^Q$, and put
\[
 c_Q(h)=\sum_{a\in Q}d_a^Qh_a.
\]
The five lifts lie in an affine three-dimensional hyperplane exactly
when $c_Q(h)=0$. These are finitely many hyperplanes in height space.
A generic perturbation inside the nonempty open set $\mathcal A$
avoids them all. Every lower facet then has four vertices and projects
to a tetrahedron; since every label remains a lower vertex, this gives
a full strictly regular tetrahedralization.

Now let $T^0,T^1$ be two such tetrahedralizations, realized by heights
$h^0,h^1$. Their realization conditions are finitely many strict
supporting-plane inequalities, so their lifting chambers are open.
Choose the heights within these chambers so that the segment
\[
 h(t)=(1-t)h^0+th^1,\qquad 0\leq t\leq1,
\]
has endpoints off the circuit hyperplanes and meets no intersection
of two of them. Such a choice avoids finitely many proper algebraic
conditions on the endpoints. Convexity keeps the entire segment in
$\mathcal A$, so no original label can disappear. Only auxiliary heights
change; all coordinates of $P$ stay fixed.

Between crossings the lower triangulation is constant. At a crossing
that changes it, a lower face contains exactly five labels: six would
force two distinct circuit equations to vanish simultaneously. The
neighboring triangulations contain its complete opposite circuit sides
and agree outside that face. A $1\mid4$ event would put the interior
label at the height of a convex combination of the other four,
contradicting \eqref{eq:activity}. Thus every change is a legal
$2\leftrightarrow3$ flip by Lemma~\ref{lem:circuit}. Ignoring crossings
that do not affect the lower hull gives the required finite path from
$T^0$ to $T^1$. Its states are the adjacent full regular triangulations;
the flat subdivisions at crossing times are not included.
\end{proof}

\begin{definition}
For $n\geq5$, a hull vertex $q$ is \emph{placing-removable} in $T$ if its
antistar has underlying space $K=\conv(P\setminus\{q\})$ and the other
tetrahedra are precisely $qF$ for the facets $F$ of $K$ visible from $q$.
We also say that $T$ is in \emph{placing form at $q$}. For $n=4$, the
sole tetrahedron is declared placing at each vertex; its antistar is the
opposite triangular face.
\end{definition}
Visibility means that $q$ lies strictly on the side opposite $K$ of the
facet's plane. For $n\geq5$, $K$ is three-dimensional and has a uniquely
triangulated boundary, since its facets are triangles. Convexity gives
\[
 \conv(P)=K\ \cup\!\!\bigcup_{F\text{ visible}}qF,
\]
with disjoint interiors and common-face intersections. This can be seen
by following the segment from $q$ to a point of $K$ to its first
intersection with $K$. If the antistar of a full triangulation has
underlying space $K$, its interface with the $q$-star is exactly this
visible boundary; hence it is in placing form.

\begin{lemma}[Regular placing extension and gluing]\label{lem:placing}
Suppose $n\geq5$ and $\mathcal F(P\setminus\{q\})$ is connected. Every
placing triangulation at $q$ is connected in $\mathcal F(P)$ to a full
strictly regular triangulation.
\end{lemma}
\begin{proof}
Let $R$ be a full regular triangulation of $P\setminus\{q\}$ with strict
lower heights $h$. Extend it by the visible cones $qF$. Give $q$ height
$M$. For sufficiently large $M$, every supporting plane of a tetrahedron
of $R$ lies strictly below $(q,M)$. The plane through $(q,M)$ and the
three lifted vertices of a visible facet $F$ lies strictly below every
other lifted point of $K$: its value at such a point decreases without
bound as $M$ grows, because the affine coordinate of $q$ relative to
$qF$ is negative there. There are finitely many inequalities, so a single
$M$ satisfies all. These tetrahedra therefore form a full strictly regular
placing extension of $R$.

By the hypothesis, the initial antistar can be flipped to $R$. Write
this smaller path as $R_0,\ldots,R_s=R$, and let $\mathcal C$ be the
fixed family of visible cones $qF$. The boundary of every $R_i$ consists
of the same triangular facets of $K$: general position leaves no facet
with a choice of boundary triangulation. Thus $R_i\cup\mathcal C$ is a
full ambient triangulation, with the same interface at every step.

Each smaller flip has its whole support inside $K$. Its complete old
circuit side is present in $R_i\cup\mathcal C$, and its replacement
changes precisely the tetrahedra of $R_i$ involved in that flip. No
member of $\mathcal C$ can belong to either circuit side, since those
sides use only labels in $P\setminus\{q\}$. Lemma~\ref{lem:circuit}
therefore gives exactly $R_{i+1}\cup\mathcal C$ as the next legal ambient
state. All original labels, including $q$, remain active. The auxiliary
notation $P\setminus\{q\}$ describes a subcomplex throughout, not a
vertex deletion along the path.
\end{proof}

Different starting triangulations may reach placing form at different
hull vertices. Lemma~\ref{lem:placing} connects each endpoint to the one
full regular locus of Lemma~\ref{lem:regular}; no common pivot is needed.

\section{Radial links and a regular lifting sweep}
\label{sec:radial}
Fix a hull vertex $q$. Write $N_T(q)$ for its mesh neighbors,
$m_T(q)=|N_T(q)|$, $b(q)$ for its degree in the graph of the convex hull,
and $i_T(q)=m_T(q)-b(q)$. These numbers differ from the number of tetrahedra
incident with $q$. Translate $q$ to the origin and choose a linear
functional $\ell$ strictly positive on every other point. Set
\begin{equation}\label{eq:radial}
 \lambda_p=\ell(p),\qquad z_p=p/\lambda_p,\qquad
 \eta_p=1/\lambda_p.
\end{equation}
The plane $\ell=1$ cuts the tangent cone of the hull at $q$ in a convex
polygon $B$. For a sufficiently small positive $s$, every segment
$[q,sp]$, $p\in P$, lies in the geometric star of $q$: the finitely many
simplices disjoint from $q$ have positive distance from $q$.
The cones on the faces of $\operatorname{lk}_T(q)$ consequently cover
that tangent cone with disjoint relative interiors. Their sections form
a straight triangulation $D$ of $B$. Its boundary vertices are exactly
the hull-edge neighbors of $q$, since these edges are precisely the
extreme rays of the tangent cone. Every other mesh neighbor projects
strictly inside $B$. Such a neighbor can itself be a global hull vertex;
it need not be an interior point of $P$. All labels of $P\setminus\{q\}$
project into $B$, whether or not they are mesh neighbors. No three
projected labels are collinear: otherwise their three original points
and $q$ would be coplanar.

A triangulation of these actual projected coordinates is \emph{strictly
upper regular} if vertex heights induce it as the projection of the upper
faces of their convex hull, with every label used and every interior
fold strict. Equivalently its piecewise-affine interpolant is concave
with a strict fold across every interior edge. For these auxiliary planar
liftings we use upper hulls; negating the heights realizes the same
triangulation under the lower-hull convention of Section~\ref{sec:prelim}.
Throughout this section,
regularity refers to these coordinates, rather than a different
realization of the same abstract disk.
\begin{figure}[ht]
\centering
\begin{tikzpicture}[scale=.86,every node/.style={font=\small}]
 \begin{scope}
  \coordinate (q) at (1.5,-.5);
  \coordinate (A) at (0,3); \coordinate (B) at (4,3);
  \coordinate (C) at (2.7,4.4); \coordinate (X) at (2,3.55);
  \fill[blue!6] (0.48,1.87)--(3.2,1.87)--(2.32,2.82)--cycle;
  \draw[blue!45] (0.48,1.87)--(3.2,1.87)--(2.32,2.82)--cycle;
  \draw[gray] (q)--(A) (q)--(B) (q)--(C);
  \draw[orange!80!black] (q)--(X);
  \foreach \p in {q,A,B,C,X} \fill (\p) circle (1.6pt);
  \node[below=3pt] (qlabel) at (q) {$q=0$};
  \node[left] at (A) {$A$}; \node[right] at (B) {$B$};
  \node[above] at (C) {$C$}; \node[right] at (X) {$x$};
  \node[blue!60!black,left] at (.4,2.25) {$\ell=1$};
  \fill[orange!80!black] (1.84,2.25) circle(1.7pt);
  \node[above left=8pt and 6pt] at (1.84,2.25) {$z_x$};
  \node[align=center,below=7pt of qlabel] {Radial projection\\[2pt] $p=\lambda_pz_p$};
 \end{scope}
 \begin{scope}[shift={(7,0)}]
  \coordinate (A) at (0,1); \coordinate (B) at (4,1);
  \coordinate (C) at (2,4); \coordinate (X) at (1.8,2);
  \fill[blue!6] (A)--(B)--(C)--cycle;
  \draw[thick,blue!65!black] (A)--(B)--(C)--cycle;
  \draw[orange!80!black] (X)--(A) (X)--(B) (X)--(C);
  \foreach \p in {A,B,C} \fill[blue!65!black] (\p) circle (2pt);
  \fill[orange!80!black] (X) circle (2pt);
  \node[left] at (A) {$z_A$}; \node[right] at (B) {$z_B$};
  \node[above] at (C) {$z_C$}; \node[above right] at (X) {$z_x$};
  \node[align=center,anchor=north,blue!65!black] (boundarylegend) at (2,.55)
    {Boundary: global hull-edge neighbors};
  \node[align=center,orange!80!black,below=7pt of boundarylegend]
    {Interior: additional mesh neighbors\\[2pt] (possibly global hull vertices)};
 \end{scope}
\end{tikzpicture}
\caption{Schematic radial section. The link fills the convex tangent-cone
section; boundary status in the section is determined by adjacency to
$q$ in the global hull, not by whether the other point is globally
interior. The drawing asserts no determinant or occupancy sign.}
\label{fig:radialprojection}
\end{figure}
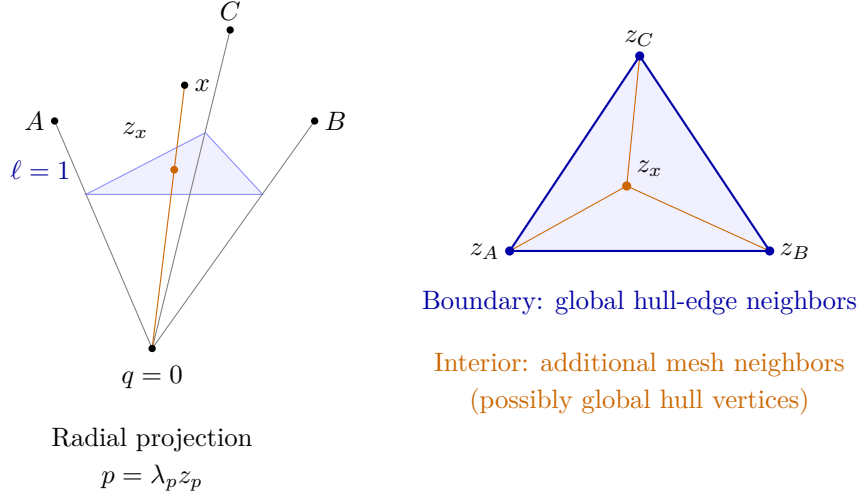

\begin{theorem}[Regular radial-link escape]\label{thm:radial}
Let $T$ be a full geometric tetrahedralization and let $q$ be a hull
vertex with $m$ mesh neighbors. If its radial link is strictly regular,
at most $\binom m4$ legal geometric $2\leftrightarrow3$ flips make $q$
placing-removable. Each move strictly shrinks its geometric star, adds
no neighbor of $q$, and preserves every initial tetrahedron not
containing $q$. If $|P|\geq5$, the endpoint admits the convex nonaffine
piecewise-affine height with $h(q)=1$ and $h(p)=0$ for $p\ne q$.
\end{theorem}
\begin{proof}
Choose a strict upper lifting $w^0$ of $D$. Its realization conditions
are finitely many strict linear inequalities, so its lifting cone is
open. Choose one circuit row for each four-label support. Such a planar
circuit has coefficients $c_p$ satisfying
\[
 \sum c_pz_p=0,\qquad \sum c_p=0.
\]
Every coefficient is nonzero. Moreover $c\cdot\eta\ne0$: otherwise
$c_p/\lambda_p$ would be an affine dependence of the corresponding four
spatial labels. Choose $w^0$ off every circuit hyperplane and every
hyperplane
\[
 (c\cdot w^0)(d\cdot\eta)=(d\cdot w^0)(c\cdot\eta)
 \quad(c\ne d).
\]
These are proper hyperplanes. In the latter equation this follows from
$c\cdot\eta,d\cdot\eta\ne0$ and the nonproportionality of circuit rows
on distinct four-label supports. An open cone cannot be covered by
finitely many proper hyperplanes.

Follow the upper hull along
\begin{equation}\label{eq:sweep}
 w(t)=(1-t)w^0+t\eta,\qquad 0\leq t\leq1.
\end{equation}
Each circuit crosses at most once, and crossing times are distinct.
Only crossings that change an upper face matter. At such a crossing
exactly four lifted labels lie on one supporting face; the adjacent
upper triangulations differ by a planar diagonal exchange, a deletion
of a degree-three interior vertex, or its reverse. Boundary events
cannot occur because there are no collinear triples. The momentary
flat upper face only describes the transition between two meshes;
it is never used as a spatial state. On the changed region the new
interpolant of the \emph{fixed} heights $\eta$ is strictly greater than
the old one: the segment crosses the circuit hyperplane toward the sign
of $c\cdot\eta$, and the two interpolants differ by a tent of constant
strict sign in the interior.

We maintain an ambient full tetrahedralization whose radial link is
the current upper triangulation. More precisely, its $q$-tetrahedra
are exactly the cones over the current upper triangles; their link
labels are a subset of the original $m$ neighbors, and every initial
non-$q$ tetrahedron is still present. All labels of $P$ remain active
in the spatial mesh, including those that cease to occur in its radial
link. These assertions hold initially. The following circuit
replacements verify them successively at each event.
For a diagonal exchange normalize $c$
to be negative on the old diagonal and positive on the new one. Increase
of the $\eta$ interpolant gives $c\cdot\eta>0$. The spatial circuit on
these four labels and $q=0$ has coefficients
\begin{equation}\label{eq:spatialcircuit}
 d_p=c_p\eta_p,\qquad d_q=-c\cdot\eta<0.
\end{equation}
Indeed $\sum d_pp=\sum c_pz_p=0$ and $d_q+\sum d_p=0$.
The two old tetrahedra containing $q$ are exactly the circuit side
obtained by omitting its two positive labels. Thus they fill the convex
five-point support, and Lemma~\ref{lem:circuit} gives a legal $2\to3$
move. Two new tetrahedra contain $q$; the third does not. The complete
old side, already part of a full geometric complex, guarantees both
support emptiness and compatibility with the exterior.

For a degree-three deletion, write
$z_x=\theta_Az_A+\theta_Bz_B+\theta_Cz_C$ with all $\theta_i>0$ and
$\sum\theta_i=1$. Increase means
\[
 \eta_x<\theta_A\eta_A+\theta_B\eta_B+\theta_C\eta_C.
\]
In the spatial circuit the positive Radon part is $\{q,x\}$, after an
overall choice of sign. The three old tetrahedra $qxAB,qxBC,qxCA$ are
its complete three-cell side. The resulting $3\to2$ replaces them by
$qABC$ and $xABC$. In particular the deleted \emph{link} label remains
an active spatial vertex in the antistar.

A planar insertion is impossible. It would require an unused link
label $x$ inside a current triangle $ABC$ with
$\eta_x>\sum\theta_i\eta_i$. But then
\[
 p_x=\sum_{i=A,B,C}\frac{\theta_i\eta_i}{\eta_x}p_i,
 \qquad 0<\sum_i\frac{\theta_i\eta_i}{\eta_x}<1.
\]
The positive remainder is the coefficient at $q$, so the already active
spatial label $p_x$ would be strictly inside the present tetrahedron
$qABC$. This contradicts the full geometric complex. The same argument
excludes reactivation of a neighbor removed at an earlier event. The
induction therefore uses only legal spatial moves, never introduces a
neighbor, and never removes an original label. Each old side consists
entirely of $q$-tetrahedra and transfers positive volume to one new
non-$q$ tetrahedron. Consequently the geometric star strictly shrinks
and every initial exterior tetrahedron remains. There are at most
$\binom m4$ events.

It remains to establish the global endpoint, rather than merely a
relative lifting. At $t=1$ let $g$ be the positive concave interpolant of
$\eta$ on the remaining link. If $A_\sigma$ is the affine supporting
plane of a triangle, then $g=\min_\sigma A_\sigma$ on $B$. Extend these
planes homogeneously to linear functions $L_\sigma$ on the tangent
cone and put
\[
 G(x)=\ell(x)g(x/\ell(x))=\min_\sigma L_\sigma(x),\qquad G(q)=0.
\]
On the ray $rz$, $z\in B$, the geometric star ends at $r=1/g(z)$.
All other active labels lie on or beyond this endpoint: a label strictly
before it would lie inside a simplex containing $q$, which is
incompatible with the complex unless it were a vertex, and even a
neighbor lies exactly at the endpoint. Thus the geometric antistar is
\begin{equation}\label{eq:antistar}
 \operatorname{conv}(P)\cap\{G\geq1\}
 =\operatorname{conv}(P)\cap\bigcap_\sigma\{L_\sigma\geq1\}.
\end{equation}
The equality also holds on the common link surface, not only away
from it. On a $q$-tetrahedron, $G=L_\sigma$ is the sum of the three
barycentric coordinates at its non-$q$ vertices. Its level $G=1$ is
therefore precisely the face opposite $q$, which belongs to the
antistar. On a tetrahedron not containing $q$, concavity of $G$ and its
values at the vertices give $G\geq1$ throughout. These observations,
together with coverage by the full mesh, prove both inclusions in
\eqref{eq:antistar}. That set is convex and contains
$P\setminus\{q\}$. Conversely every simplex in the antistar has vertices in $P\setminus\{q\}$.
These two inclusions show that it is exactly
$\operatorname{conv}(P\setminus\{q\})$.

On a $q$-tetrahedron the unit-$q$ interpolant is the barycentric
coordinate $1-L_\sigma=1-G$, and on the antistar it is zero. Globally it
therefore equals
\begin{equation}\label{eq:unitheight}
 f(x)=\max(0,1-G(x))
     =\max\bigl(0,\{1-L_\sigma(x)\}_\sigma\bigr).
\end{equation}
This maximum of affine functions is convex. If $|P|\geq5$, four points
other than $q$ are affinely independent. No affine function can vanish
at all four and equal one at $q$, so this height is nonaffine.
\end{proof}

\section{Forest liftings and connectivity through seven points}
\label{sec:forest}
The following lifting criterion is useful because it concerns only the
graph induced by the interior vertices of a planar disk. We give the
complete equilibrium construction; no novelty is asserted for this
planar criterion. It is an instance of the equilibrium--lifting
correspondence discussed in the standard theory of regular
triangulations~\cite{DRS2010}.

\begin{theorem}[Interior-forest regularity]\label{thm:forest}
Let $D$ be a straight triangulation of a planar point set with no three
collinear, whose hull is a convex polygon. If the graph induced by its
interior vertices is a forest, then $D$ is strictly regular in its given
coordinates.
\end{theorem}
\begin{proof}
At each interior vertex $v$, the vectors $u-v$ to all its neighbors
surround the origin. Hence choose strictly positive coefficients
$a_{vu}$ such that
\[
 \sum_{u\sim v}a_{vu}(u-v)=0.
\]
One direct construction assigns a small positive coefficient to every
vector and cancels their sum with a nonnegative combination of the
same vectors; their positive cone is the whole plane. Root each tree
of interior vertices. Starting with any positive root scale, define
recursively
\[
 t_u=t_va_{vu}/a_{uv}
\]
at each interior child $u$ of $v$. A forest has no competing path, so
these scales are consistent. Put $\omega_{vu}=t_va_{vu}$; it is symmetric
on edges with two interior endpoints. On an edge from an interior to a
boundary vertex it is defined by its interior endpoint. Give arbitrary
positive weights to interior edges with two boundary endpoints. The
result is a symmetric positive weight on every nonboundary edge, with
\begin{equation}\label{eq:equilibrium}
 \sum_{u\sim v}\omega_{vu}(u-v)=0
 \quad\hbox{at each interior vertex }v.
\end{equation}

To integrate a lifting, write the desired affine height function on a
triangle $\tau$ as
\[
 F_\tau(x)=g_\tau\cdot x+b_\tau,
\]
where $g_\tau\in\mathbb R^2$ is its gradient and $b_\tau\in\mathbb R$
its constant term. We first prescribe the differences between neighboring
triangle planes; consistency around closed paths will be checked below.
Let $J$ be counterclockwise rotation by $\pi/2$. For an oriented interior
edge $v\to u$ with left and right incident triangles $L,R$, prescribe
\begin{equation}\label{eq:plane-jumps}
 g_R-g_L=\omega_{vu}J(u-v),\qquad
 b_R-b_L=-\omega_{vu}J(u-v)\cdot v.
\end{equation}
The gradient difference is perpendicular to the shared edge, so the two
planes have the same slope along it. To make their heights agree at $v$
as well, we require
$(g_R-g_L)\cdot v+(b_R-b_L)=0$, which gives the second equation.
Without that equation the planes could retain a vertical mismatch.
Together the two prescriptions give
\[
 F_R(x)-F_L(x)=\omega_{vu}J(u-v)\cdot(x-v).
\]
For every point $x=v+s(u-v)$ of the shared edge, the right side is zero
because $J(u-v)\cdot(u-v)=0$. Thus the neighboring height functions agree
along the entire edge. The same orthogonality gives
$J(u-v)\cdot u=J(u-v)\cdot v$, so either endpoint can be used in the
constant-term equation. Reversing the edge orientation exchanges $L,R$
and reverses both differences; the prescription is therefore independent
of this orientation choice.

The sum around a small dual loop about an interior vertex is, up to the
orientation sign, $J\sum_u\omega_{vu}(u-v)=0$. Its constant-term sum is
zero as well, since all crossed edges contain the same vertex and it is
the negative scalar product of that vertex with the summed gradient
jump. These loops span the cycle space of the triangle-adjacency dual
graph. Indeed, writing $v_D,e_D,f_D$ for the numbers of vertices, edges,
and triangles, and $b$ for the boundary size, its dimension is
\[
 (e_D-b)-f_D+1=v_D-b,
\]
by Euler's formula. This is exactly the number of interior vertices.
The displayed loops are independent: a linear relation gives equal
coefficients at the endpoints of each interior--interior edge and
coefficient zero at the interior endpoint of an interior--boundary
edge. Each interior component has a path to the boundary in the
triangulation's connected graph, so all coefficients vanish.
Consequently, after choosing a gradient and constant term on one triangle,
the prescribed differences assign them to every other triangle
independently of the dual path. The resulting affine functions agree on
shared edges and hence give a continuous piecewise-affine function on $D$.

A transverse direction crossing $v\to u$ from left to right has negative
scalar product with $J(u-v)$. Equation~\eqref{eq:plane-jumps} thus strictly
decreases the directional slope. Denote the resulting continuous height
function by $F$. On a line segment $x(s)=x_0+sd$ that avoids vertices
and crosses edges transversely, $F(x(s))$ is piecewise linear: its slope
inside a triangle $\tau$ is $g_\tau\cdot d$, and every edge crossing
strictly decreases that slope. A continuous piecewise-linear function
with nonincreasing slopes is concave: for $a<b<c$, its average slope
on $[a,b]$ is at least its average slope on $[b,c]$. Thus the concavity
inequality holds on
each such segment. Any segment in the convex polygon can be approximated
by these generic segments with endpoints in its interior. Continuity
then gives, for arbitrary points $x,y$ of the polygon,
\[
 F((1-s)x+sy)\geq(1-s)F(x)+sF(y),\qquad 0\leq s\leq1.
\]
This proves global concavity, including along segments through vertices
or along edges. Every interior fold is strict because its weight
$\omega_{vu}$ is positive; $F$ remains affine within each triangle.

Lift each vertex $v$ to $(v,F(v))$. Concavity implies that every convex
combination of lifted vertices lies on or below the graph of $F$:
\[
 \sum_i\alpha_iF(v_i)\leq F\!\left(\sum_i\alpha_i v_i\right)
 \qquad\left(\alpha_i\geq0,\quad\sum_i\alpha_i=1\right).
\]
Conversely, a point $x$ in a triangle is a convex combination of its
three vertices, and affinity on that triangle gives the same combination
for $F(x)$. Hence $(x,F(x))$ belongs to the lifted convex hull. These
two observations show that its upper surface is exactly the graph of $F$.

It remains to check that its maximal upper faces are exactly the lifted
triangles. The affine plane over any triangle $\tau$ supports $F$ from
above: apply concavity along a segment starting in the interior of
$\tau$, where $F$ agrees with that plane. The set where equality holds
is convex, since equality at two points forces equality along their
segment. If this contact set contained a point outside $\tau$, its
convex hull with $\tau$ would extend the same plane across an interior
edge of $\tau$, contradicting the strict fold there. Thus each lifted
triangle is a maximal upper face, and the projected upper triangulation
is precisely $D$, with every vertex used.
\end{proof}

\begin{lemma}[Auxiliary planar ear removal]\label{lem:ear}
Suppose an interior vertex $x$ of a planar triangulation has degree
three, with neighbors $A,B,C$. Replacing $xAB,xBC,xCA$ by $ABC$ gives a
straight triangulation $D^-$. Then $D$ is strictly regular if and only
if $D^-$ is strictly regular.
\end{lemma}
\begin{proof}
The three angles around $x$ fill a neighborhood, so $x$ is strictly
inside $ABC$ and its three triangles fill that triangle. Starting with
a strict lifting of $D^-$, give $x$ a height $\epsilon>0$ above the
plane over $ABC$. The three inner folds become strictly concave. All
other strict inequalities persist for sufficiently small $\epsilon$,
so this is a strict lifting of $D$.

Conversely choose a generic strict upper lifting of $D$, avoiding every
four-label coplanarity. Remove the lifted point $x$. Every old upper
triangle not containing $x$ remains an upper triangle: its supporting
plane still lies above all remaining labels. The only projected region
left to cover is $ABC$, containing no other label. It is consequently
one upper triangle. Genericity excludes a flat fold along its boundary,
so the resulting upper triangulation is precisely $D^-$ and is strict.
\end{proof}
This operation is used only to construct planar heights. It is not a
spatial vertex deletion, contraction, or allowed edge of the flip graph.
In particular a degree-three spatial edge alone does not guarantee a
legal spatial $3\to2$ flip.

\Needspace{8\baselineskip}
\begin{theorem}[Connectivity through seven points]\label{thm:seven}
Every full tetrahedralization on at most seven points has a hull radial
link whose interior graph is a forest. Its fixed-label geometric flip
graph is connected.
\end{theorem}
\begin{proof}
Let $n=|P|$ and $h$ be the number of hull vertices. Suppose every hull
link has an interior cycle. Then $i_T(q)\geq3$ at every hull vertex,
whereas $m_T(q)\leq n-1$, giving $b(q)\leq n-4$. The simplicial hull has
$3h-6$ edges, so
\begin{equation}\label{eq:hdegree}
 6h-12=\sum_{q\text{ hull}}b(q)\leq h(n-4),\qquad (10-n)h\leq12.
\end{equation}
For $n\leq6$ this contradicts $h\geq4$. For $n=7$ it forces $h=4$.
All three global interior labels, say $a,b,c$, must then be interior
neighbors at every hull pivot, and their cycle is the triangle $abc$.
No projected boundary vertex can lie inside this triangle. There is no
other interior label to triangulate it, so it is a face of every hull
link. The four hull vertices would therefore give four tetrahedra with
the same triangular face $abc$. A triangular face in a geometric
three-dimensional triangulation has at most two tetrahedral cofaces,
a contradiction.

The forest criterion and Theorem~\ref{thm:radial} reach placing.
Induct on $n$, with the unique tetrahedron at $n=4$ as base case.
The placing antistar is a full tetrahedralization of
$\operatorname{conv}(P\setminus\{q\})$. By induction it has a flip path
to a full regular tetrahedralization. The fixed triangular hull
boundary of that smaller convex hull is unchanged, so this path glues
to the fixed cones from $q$ over its visible facets. Lemma~\ref{lem:placing} supplies a regular placing extension. Thus every state
reaches the full regular locus, which is connected by Lemma~\ref{lem:regular}. Throughout the
glued path $q$ remains active; the smaller point set is only an auxiliary
object in the proof.
\end{proof}

\section{Eight points with at least five hull vertices}
\label{sec:hfiveplus}
\begin{theorem}\label{thm:hfiveplus}
For eight points with $h\in\{5,6,7,8\}$ hull vertices, every full
tetrahedralization has a strictly regular actual hull radial link.
It reaches placing in at most $35$ flips and belongs to the connected
full regular component.
\end{theorem}
\begin{proof}
If $h=7$ or $8$, the hull degree sum $6h-12$ supplies a vertex with
$b(q)\geq5$. Since $m_T(q)\leq7$, its radial link has at most two interior
vertices. Its interior graph is a forest, so it is strictly regular.

If $h=5$, the simplicial hull is a triangular bipyramid. Each of its
three equatorial vertices has four hull neighbors. If all three of
these links failed the forest criterion, each would contain the same
three global interior labels $a,b,c$ and their triangular cycle. As in
Theorem~\ref{thm:seven}, this cycle must be a link face. The three
actual tetrahedra with apices at the equatorial vertices would give
three cofaces of $abc$, which is impossible. Thus a forest link exists.

For $h=6$, write $x,y$ for the global interior labels and suppose all
hull links are nonregular. The forest criterion gives $i_T(q)\geq3$ and
$b(q)\leq4$ at every hull vertex. Their degrees sum to $24$, so all are
four. The complement of this six-vertex graph is a perfect matching;
only now may we conclude that the hull is octahedral. Let $q'$ denote
the vertex opposite $q$. The only possible interior radial-link labels
at $q$ are $x,y,q'$. Nonregularity forces their triangular cycle and,
since there are no other interior link vertices, their triangular face.
Thus the three tetrahedra $qq'xy$, one for each opposite pair, are
actual tetrahedra of $T$.

The link of the interior point $x$ is a triangulated two-sphere on $y$
and the six hull labels. The displayed tetrahedra imply that $y$ is
adjacent in this sphere to every hull label. Hence its link is a
six-cycle, and the complementary closed disk is a triangulated
hexagon with no interior vertices. A triangulated polygon has an ear:
its triangle-adjacency graph is a tree, whose leaf gives a triangle
with two boundary edges. Let $q$ be the tip of such an ear. Exactly
one triangle of the complementary disk and two triangles of the
$y$-star meet at $q$. Hence $q$ has degree three in
$\operatorname{lk}_T(x)$; equivalently the edge $xq$ belongs to exactly
three tetrahedra and $x$ has degree three in the radial link at $q$.

Apply Lemma~\ref{lem:ear} to remove $x$ in that planar lifting
construction. The remaining disk has only the two interior labels
$y,q'$ and is regular by Theorem~\ref{thm:forest}. Reinsertion gives a
strict lifting of the original actual link, contradicting its assumed
nonregularity. This reasoning has used no spatial edge contraction.

In each case $m_T(q)\leq7$, so Theorem~\ref{thm:radial} gives at most
$\binom74=35$ flips to placing. The seven-point antistar induction and
regular placing extension then reach the connected full regular locus.
\end{proof}

\begin{table}[ht]
\centering\small
\begin{tabular}{c p{0.51\linewidth} c}
\hline
Hull vertices & Reason a regular link exists & Flips to placing\\
\hline
$8,7$ & A hull degree at least five leaves at most two interior link vertices. & $\leq35$\\
$5$ & Three equatorial cyclic links would give three cofaces of one face. & $\leq35$\\
$6$ & All nonregular links force an octahedral hull; its universal link neighbor supplies an auxiliary ear. & $\leq35$\\
\hline
\end{tabular}
\caption{The eight-point cases before the tetrahedral hull. All rows use
forest lifting and regular radial escape; connectivity then uses the
self-contained seven-point result and the full regular locus.}
\label{tab:hfiveplus}
\end{table}

\section{Positive-target preparation for a triangular radial link}
\label{sec:octprep}

A tetrahedral hull is the only distribution still requiring an argument.
Its radial links have triangular boundary. With three interior link
vertices, the possible obstruction is an octahedral disk; with four,
compatibility between different hull links will supply a regular one.
The first step uses a geometric feature specific to a tetrahedral hull:
the reciprocal target lies strictly above its boundary affine plane.

\begin{lemma}[Positive gauged target]\label{lem:positive-target}
Let the hull be $qABC$, and use the radial coordinates
$z_p=p/\lambda_p$, $\eta_p=1/\lambda_p$ after translating $q$ to zero.
If $L$ is affine on the section and takes the values
$\eta_A,\eta_B,\eta_C$ at its boundary vertices, then
\[
                 \eta_x-L(z_x)>0
\]
for every interior label $x$ adjacent to $q$.
\end{lemma}
\begin{proof}
Write $z_x=\theta_Az_A+\theta_Bz_B+\theta_Cz_C$, where the
$\theta$'s are positive and sum to one. Then
\[
 x=\sum_{v\in\{A,B,C\}}\frac{\theta_v\eta_v}{\eta_x}\,v.
\]
Put $\alpha_v=\theta_v\eta_v/\eta_x$ for $v\in\{A,B,C\}$.
Since $q=0$ and $q,A,B,C$ are affinely independent, these are exactly
the coefficients of $A,B,C$ in the spatial barycentric expression for
$x$. Its remaining coefficient is
$\alpha_q=1-\sum_v\alpha_v$. The label $x$ is strictly inside the
tetrahedral hull, so all four spatial coefficients, in particular
$\alpha_q$, are positive. Affineness of $L$ and
$\sum_v\theta_v=1$ now give the exact gap
\[
 \eta_x-L(z_x)
 =\eta_x-\sum_v\theta_v\eta_v
 =\eta_x\Bigl(1-\sum_v\alpha_v\Bigr)
 =\frac{\alpha_q}{\lambda_x}>0.
\]
Thus the positive gap records the spatial barycentric weight at the
pivot, divided by the positive radial scale of $x$. The coefficients
$\theta_v$ describe the section point, whereas the $\alpha_v$ describe
the original spatial point; the displayed rescaling is what relates
these two convex combinations. Notice that the assertion concerns
$\eta-L$, rather than positivity of $\eta$ alone.
\end{proof}

We next record the lifting calculation with all sign conditions explicit.
The \emph{octahedral disk} has boundary $e_0e_1e_2$ and interior vertices
$u_0,u_1,u_2$; its only missing boundary--interior edges are $e_i u_i$.
Its triangles are the seven choices of one vertex from each pair
$\{e_i,u_i\}$ other than the outer triangle.

Use the three projected boundary vectors as a linear basis. In these
coordinates $e_0,e_1,e_2$ are the standard coordinate vectors and the
radial section is $X_0+X_1+X_2=1$. Write the barycentric coordinates of
the $j$th interior point as a column $\beta_j=(\beta_{0j},\beta_{1j},
\beta_{2j})^{\mathsf T}$, with all entries positive and their sum one.
A positive multiple of this column represents the same ray from $q$;
dividing by its coordinate sum recovers the same point on the section.
This freedom to rescale a representative is what we mean by
\emph{homogeneous coordinates}. Divide column $j$ by its positive entry
$\beta_{jj}$ and denote the resulting representative by $u_j$. Then
\[
 U=(u_0\ u_1\ u_2)=
 \begin{pmatrix}1&a&b\\c&1&d\\e&f&1\end{pmatrix},
 \qquad a,b,c,d,e,f>0.
\]
The diagonal entries are one by construction. The columns generally do
not lie on the radial section: if $s_j=\sum_iU_{ij}=1/\beta_{jj}$, the
actual projected point has coordinates $u_j/s_j=\beta_j$. Thus this
normalization changes only the representatives, not the projected
points or the original spatial labels. Positive column rescaling also
preserves the signs of the orientation determinants used below.

Lifting heights must be rescaled by the same factors. If $H_j$ is the
height at the actual projected point, its normalized height is
$h_j=s_jH_j$: the entire lifted representative is
$(u_j,h_j)=s_j(\beta_j,H_j)$. Dividing by $s_j$ recovers both the point
and its height. An affine height function on the section extends to a
linear form in these coordinates, so its supporting inequalities can
be evaluated on the normalized representatives. This makes the cap
planes in the next proof particularly simple: with boundary heights
zero, the plane through $e_1,e_2,u_0$ has height $h_0X_0$, since its
first coordinate is zero at $e_1,e_2$ and one at $u_0$.

Put
\[
 \rho=ac-1,\quad \sigma=be-1,\quad \tau=df-1,\qquad
 C_+=ade,\quad C_-=bcf.
\]
\[
 \Delta=\det U=1+C_++C_--ac-be-df.
\]
The orientations of the three annular triangles
$(e_2,u_1,u_0)$, $(e_1,u_0,u_2)$, $(e_0,u_2,u_1)$, and of the
central triangle $(u_0,u_1,u_2)$ give
\begin{equation}\label{eq:octorient}
                  \rho,\sigma,\tau,\Delta>0.
\end{equation}
Indeed, orient the section so that the boundary order $e_0,e_1,e_2$
is positive. The displayed orders of the four interior triangles are
the compatible positive orders in the actual disk: adjacent triangles
traverse their shared edge in opposite directions. On the section,
the determinant of three columns has the sign of their planar
orientation. The same remains true for our representatives because
each column was rescaled by a positive factor. Expanding the four
determinants gives respectively $ac-1$, $be-1$, $df-1$, and $\det U$.
Their positivity therefore uses the noncrossing geometric realization
of the disk, not just positivity of the six entries of $U$.

\begin{lemma}[Octahedral lifting criterion]\label{lem:octregular}
For the actual projected coordinates of an octahedral disk, strict upper
regularity is equivalent to $C_+>1$ and $C_->1$.
\end{lemma}
\begin{proof}
Subtract the affine boundary height, making the outer heights zero, and
scale heights along with their homogeneous representatives. Write $h_i$
for the inner heights. In a strict concave lifting, each $h_i>0$.
The plane on the cap opposite $e_i$ is $h_i$ times coordinate $i$.
Its strict supporting inequalities at the other inner vertices are
\begin{equation}\label{eq:capineq}
                         U_{ij}h_i>h_j\qquad(i\ne j).
\end{equation}
For example, the inequalities
$ah_0>h_1$, $dh_1>h_2$, and $eh_2>h_0$ multiply to
$ade\,h_0h_1h_2>h_0h_1h_2$. Cancelling the positive product of
heights gives $C_+>1$. Multiplying the other three inequalities,
$bh_0>h_2$, $fh_2>h_1$, and $ch_1>h_0$, gives $C_->1$.

For sufficiency choose positive numbers $\mu_0,\mu_1,\mu_2$ and set
$h=U^{\mathsf T}\mu$. The central triangle then has supporting plane
$\mu\cdot z$, strictly above each outer lifted point. The six conditions
\eqref{eq:capineq} become
\begin{align*}
 \rho\mu_1+(ae-f)\mu_2&>0,&
 \sigma\mu_2+(bc-d)\mu_1&>0,\\
 \rho\mu_0+(cf-e)\mu_2&>0,&
 \tau\mu_2+(ad-b)\mu_0&>0,\\
 \sigma\mu_0+(de-c)\mu_1&>0,&
 \tau\mu_1+(bf-a)\mu_0&>0.
\end{align*}
An inequality whose second coefficient is nonnegative is automatic.
Each other inequality prescribes a positive strict ratio bound between
two $\mu$'s. Such bounds have a simultaneous positive solution exactly
when the product of the lower-bound factors on every directed cycle is
less than one. Here is a direct sufficiency proof of that elementary
fact. Enlarge every factor by the same number $\delta>1$, sufficiently
close to one that all simple-cycle products remain below one. Introduce
a source with an edge of factor one to each vertex. At each vertex take
the largest product of factors along a path from the source. Removing
a cycle increases that product, so a maximum is attained on a simple
path. These maxima satisfy the enlarged weak bounds, hence the original
strict bounds.

In the present three-vertex ratio graph, the possible two-cycle tests are
exactly
\begin{align*}
 \rho\sigma-(f-ae)(d-bc)&=\Delta,\\
 \rho\tau-(e-cf)(b-ad)&=\Delta,\\
 \sigma\tau-(c-de)(a-bf)&=\Delta.
\end{align*}
The two possible three-cycle tests are
\begin{align*}
 \rho\sigma\tau-(a-bf)(d-bc)(e-cf)&=(C_--1)\Delta,\\
 \rho\sigma\tau-(b-ad)(c-de)(f-ae)&=(C_+-1)\Delta.
\end{align*}
A test is needed only when all its numerator factors are positive.
The displayed identities and \eqref{eq:octorient} make every needed
test strict when $C_+,C_->1$. Thus suitable positive $\mu$'s exist.
The central triangle and all three caps are strict upper faces. For
the central plane this follows from $\mu_i>0$ at each outer vertex;
for a cap, \eqref{eq:capineq} gives strictness at the other inner
vertices, and $h_i>0$ gives strictness at the opposite outer vertex.
Thus each of these four supporting planes contains just the three
lifts defining its triangle. Their projections are the corresponding
four triangles of the specified geometric disk.

The upper hull projects onto the whole outer triangle. Triangulate any
of its remaining faces using their vertices, keeping the four strict
triangular faces. Their complement in the specified disk consists of
three triangular regions with no further vertices. Their boundary
edges are already edges of the four fixed faces, so an upper
triangulation cannot cross these boundaries. It must fill each region
by its unique triangle. Every edge of a remaining triangle borders
one of the four strict faces. Across such an edge, the other triangle's
third vertex lies strictly below that face's plane; hence the two
triangles cannot be coplanar and their upper fold is strict. This also
shows that no nontriangular upper face was subdivided by a flat edge.
The resulting upper triangulation is exactly the specified disk.
\end{proof}

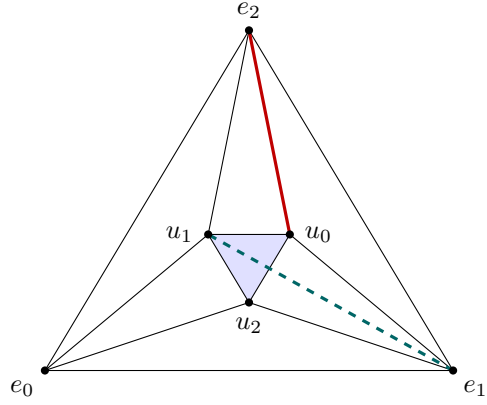
\begin{figure}[htbp]\centering
\begin{tikzpicture}[scale=0.90,every node/.style={font=\small},vertex label/.style={inner sep=1.2pt,outer sep=0pt}]
\coordinate (A) at (0.00000,0.00000);
\coordinate (B) at (6.00000,0.00000);
\coordinate (C) at (3.00000,5.00000);
\coordinate (a) at (3.60000,2.00000);
\coordinate (b) at (2.40000,2.00000);
\coordinate (c) at (3.00000,1.00000);
\fill[blue!12] (a)--(b)--(c)--cycle;
\draw (A)--(B);
\draw (A)--(C);
\draw (A)--(b);
\draw (A)--(c);
\draw (B)--(C);
\draw (B)--(a);
\draw (B)--(c);
\draw (C)--(a);
\draw (C)--(b);
\draw (a)--(b);
\draw (a)--(c);
\draw (b)--(c);
\draw[red!75!black,very thick] (C)--(a);
\draw[teal!80!black,dashed,very thick] (B)--(b);
\fill (A) circle (1.7pt);
\node[vertex label,anchor=north east] at ([xshift=-3pt,yshift=-3pt]A) {$e_0$};
\fill (B) circle (1.7pt);
\node[vertex label,anchor=north west] at ([xshift=3pt,yshift=-3pt]B) {$e_1$};
\fill (C) circle (1.7pt);
\node[vertex label,anchor=south] at ([xshift=0pt,yshift=4pt]C) {$e_2$};
\fill (a) circle (1.7pt);
\node[vertex label,anchor=west] at ([xshift=5pt,yshift=0pt]a) {$u_0$};
\fill (b) circle (1.7pt);
\node[vertex label,anchor=east] at ([xshift=-5pt,yshift=0pt]b) {$u_1$};
\fill (c) circle (1.7pt);
\node[vertex label,anchor=north] at ([xshift=0pt,yshift=-5pt]c) {$u_2$};
\end{tikzpicture}
\caption{Schematic octahedral radial disk. A cycle exchange replaces the old cross edge $e_2u_0$ (solid red) by $e_1u_1$ (dashed green). The determinant and target inequalities in the proof, not the drawing, establish convexity and direction.}\label{fig:octahedral}
\end{figure}

\Needspace{8\baselineskip}
\begin{lemma}[Positive-target exchange]\label{lem:octexchange}
If an octahedral disk is not strictly regular, every positive inner
height vector with zero boundary heights and nonzero circuit values
admits a geometrically convex diagonal exchange that strictly raises its
piecewise-affine interpolant. The new disk is strictly regular.
\end{lemma}
\begin{proof}
By Lemma~\ref{lem:octregular}, one cyclic product is at most one.
Index that cycle by distinct $i,j,k$, so that
$U_{ij}U_{jk}U_{ki}\leq1$. We first show that each of its three cross
edges $e_k u_i$ is geometrically exchangeable. The pair inequalities
in \eqref{eq:octorient} give $U_{kj}U_{jk}>1$, hence
\[
 U_{kj}>1/U_{jk}\geq U_{ij}U_{ki},
\]
where the second inequality follows by dividing the cyclic-product
inequality by $U_{jk}$. Put
$\kappa=U_{ji}U_{ij}-1>0$ and
$\nu=U_{kj}-U_{ki}U_{ij}>0$. The coordinate identity
\begin{equation}\label{eq:oct-radon}
 u_j+(U_{ji}U_{ij}-1)e_j
   =U_{ij}u_i+(U_{kj}-U_{ki}U_{ij})e_k
\end{equation}
can be checked coordinate by coordinate: both sides have entries
$U_{ij},U_{ji}U_{ij},U_{kj}$ in coordinates $i,j,k$.

To see the geometry of this positive relation, write
$\widehat u_r=u_r/s_r$ for the actual section point, where
$s_r=\sum_\ell U_{\ell r}$. Taking coordinate sums in
\eqref{eq:oct-radon} gives
$M=s_j+\kappa=U_{ij}s_i+\nu>0$, and division by $M$ gives
\[
 \xi=\frac{s_j\widehat u_j+\kappa e_j}{M}
     =\frac{U_{ij}s_i\widehat u_i+\nu e_k}{M}.
\]
Each expression is a strict convex combination of its two endpoints.
Thus the segments $\widehat u_je_j$ and $\widehat u_ie_k$ cross in
their interiors and are the diagonals of a convex quadrilateral.
The latter is the old diagonal. This verifies geometric exchangeability
for each cyclic choice of $i,j,k$, even when the cyclic product is one.

We now select an exchange using the given heights. Recall that the
normalized heights are $h_r=s_rH_r>0$. The three inequalities
\[
 h_j<U_{ij}h_i,\qquad h_k<U_{jk}h_j,\qquad h_i<U_{ki}h_k
\]
cannot all hold: multiplying and cancelling $h_ih_jh_k>0$ would give
$1<U_{ij}U_{jk}U_{ki}\leq1$. Equality in an individual comparison
would also be impossible. Indeed, the height plane through
$e_j,e_k,u_i$ is $h_iX_i$, so $h_j=U_{ij}h_i$ would place the fourth
lift $u_j$ on that plane, giving a zero circuit value. After a cyclic
reindexing we therefore have $h_j>U_{ij}h_i$.

Replace the old triangles $e_je_ku_i,e_ku_iu_j$ by
$e_je_ku_j,e_ju_iu_j$, which use the new diagonal $e_ju_j$.
Let $g_{\rm old}$ and $g_{\rm new}$ interpolate the same actual target
heights $H_r$ and the zero boundary heights on these two triangulations.
At the diagonal intersection above,
\[
 g_{\rm new}(\xi)-g_{\rm old}(\xi)
       =\frac{h_j-U_{ij}h_i}{M}>0.
\]
Their difference is affine on each of the four triangles obtained by
drawing both diagonals, and vanishes on the quadrilateral boundary.
Each refinement triangle has $\xi$ as one vertex and one side of the
quadrilateral as its opposite edge. Its affine difference is therefore
the positive value at $\xi$ multiplied by the barycentric coefficient
of $\xi$. That coefficient is positive away from the opposite boundary
edge, including the portions of both diagonals inside the quadrilateral.
Consequently the difference is strictly positive throughout the
quadrilateral interior. Outside the quadrilateral the interpolants agree.

The exchange removes $e_ku_i$ and adds $e_ju_j$. Thus $u_i$, whose
old neighbors were $e_j,e_k,u_j,u_k$, now has exactly the three
neighbors $e_j,u_j,u_k$. Remove $u_i$ temporarily in the auxiliary
planar lifting problem. The remaining two interior vertices induce a
forest, so Theorem~\ref{thm:forest} supplies a strict lifting of that
smaller disk. Lemma~\ref{lem:ear} reinserts $u_i$ coherently and proves
strict regularity of the exchanged six-vertex disk. This constructs
some strict lifting; its heights need not be the given target heights.
The auxiliary removal and reinsertion change no spatial point set.
All coefficient and height comparisons used above remain strict in the
cyclic-product equality case, without perturbing the spatial coordinates.
\end{proof}

\begin{theorem}[Six-neighbor preparation]\label{thm:octprep}
Suppose the hull is $qABC$ and $q$ has at most three additional mesh
neighbors. At most sixteen legal geometric flips make $q$
placing-removable. Each move shrinks its geometric star, preserves its
initial exterior tetrahedra, and introduces no new $q$-neighbor.
\end{theorem}
\begin{proof}
The radial boundary consists of the three hull neighbors $A,B,C$,
so $q$ has at most six mesh neighbors. If the link is regular,
Theorem~\ref{thm:radial} gives at most $\binom64=15$ flips.
Otherwise Lemma~\ref{lem:cores} below shows that it has exactly three
interior vertices and is the octahedral disk.

Lemma~\ref{lem:positive-target} supplies the target $\eta-L$, which
vanishes on the radial boundary and is positive at the interior
neighbors. To check the remaining hypothesis of
Lemma~\ref{lem:octexchange}, take any four-label radial circuit
\[
 \sum_p c_pz_p=0,\qquad \sum_p c_p=0.
\]
If $\sum_p c_p\eta_p=0$, then the nonzero coefficients
$c_p\eta_p$ would sum to zero and satisfy
$\sum_p(c_p\eta_p)p=\sum_p c_pz_p=0$. They would give an affine
dependence of four spatial labels, contradicting general position.
Moreover, because $L$ is affine,
\[
 \sum_p c_p\bigl(\eta_p-L(z_p)\bigr)=\sum_p c_p\eta_p\ne0.
\]
Positive rescaling of the lifted homogeneous representatives preserves
this nonvanishing condition. Lemma~\ref{lem:octexchange} therefore
gives a convex exchange that increases the target and makes the disk
strictly regular.

We next verify that this planar exchange is supported by a spatial
$2\to3$ flip. Name its four spatial labels $v_1,v_2,v_3,v_4$, with
old diagonal $v_3v_4$ and new diagonal $v_1v_2$ in the radial section.
For this calculation abbreviate $z_{v_r}$ and $\eta_{v_r}$ by $z_r$
and $\eta_r$. Normalize \eqref{eq:oct-radon} to
\[
 \alpha_1z_1+\alpha_2z_2=\beta_3z_3+\beta_4z_4,
 \qquad \alpha_1+\alpha_2=\beta_3+\beta_4=1,
\]
with all four coefficients positive. Each side describes the diagonal
intersection. The difference between the new and old interpolated
heights there is positive. Affine gauging cancels from this difference,
so in terms of the raw reciprocals it says
\[
 \delta:=\alpha_1\eta_1+\alpha_2\eta_2
          -\beta_3\eta_3-\beta_4\eta_4>0.
\]
Substitute $z_r=\eta_rv_r$. Since $q=0$, adding $\delta q$ does not
change the vector equality, and makes the coefficient sums equal:
\[
 \alpha_1\eta_1v_1+\alpha_2\eta_2v_2
   =\beta_3\eta_3v_3+\beta_4\eta_4v_4+\delta q.
\]
All coefficients are strictly positive. This is therefore the spatial
affine dependence with Radon partition
$\{v_1,v_2\}\mid\{q,v_3,v_4\}$.

The two old radial triangles belong to the actual link, so their cones
$qv_1v_3v_4$ and $qv_2v_3v_4$ are tetrahedra of the current mesh.
They are precisely the complete circuit side obtained by omitting
$v_2$ and $v_1$, and hence fill the convex five-point support.
No other active label can lie in their union: it would be contained
in an existing simplex without being one of its vertices, contrary
to the full simplicial-complex condition. Lemma~\ref{lem:circuit}
therefore gives the legal replacement
\[
 \{qv_1v_3v_4,qv_2v_3v_4\}
 \longrightarrow
 \{v_1v_2v_3v_4,qv_1v_2v_3,qv_1v_2v_4\}.
\]
It preserves the whole support boundary and every exterior tetrahedron.
The nondegenerate tetrahedron $v_1v_2v_3v_4$ is transferred to the
antistar, so the geometric $q$-star strictly shrinks. All four affected
labels still occur in the new $q$-tetrahedra, so every neighbor is
retained and none is introduced. In particular, the auxiliary planar
removal used to prove regularity has deleted no spatial label.

The new actual link is the strictly regular six-vertex disk supplied
by Lemma~\ref{lem:octexchange}. A radial sweep now uses at most
fifteen further flips, each preserving the preparation's exterior
tetrahedra and the other asserted properties. The total is at most
$1+\binom64=16$ flips to placing.
\end{proof}

\section{Small triangular disks and four-link compatibility}
\label{sec:fourlinks}

When all four interior labels are neighbors of every hull vertex, a
single link need not admit the preceding preparation. The four disks,
however, encode the same spatial faces. We classify only the few forms
that could be nonregular and then compare those shared faces. The
classification concerns auxiliary planar liftings throughout.

\begin{lemma}[Small triangular disks]\label{lem:cores}
A nonregular triangular disk with at most three interior vertices has
exactly three, and is the octahedral disk. A nonregular triangular disk
with four interior vertices has one of five forms: pendant,
annular-stellar, central-stellar, pentagonal-bipyramid core, or sector.
Their combinatorial types alone are not assertions of nonregularity.
\end{lemma}
\begin{proof}
Repeatedly remove interior vertices of degree three in the auxiliary
planar disk. Their three incident triangles are replaced by the triangle
on their neighbors. Lemma~\ref{lem:ear} preserves regularity in both
directions, so every remainder of a nonregular disk is nonregular.
A disk with at most two interior vertices is regular by
Theorem~\ref{thm:forest}, since its interior graph is a forest.
Consequently a nonregular disk with at most three interior vertices
already has exactly three and no removable vertex. Its interior graph
must contain a cycle, again by the forest criterion. On three vertices
this is the triangle, which is a face: no boundary vertex can lie
inside it, and there is no further interior label to subdivide it.

Euler's formula for a triangulated disk with triangular boundary and
$k$ interior vertices gives $3k+3$ edges. Indeed, if $F$ counts its
triangular faces and $E$ its edges, Euler gives $(k+3)-E+F=1$;
counting face-edge incidences gives $3F=2E-3$, since just the three
boundary edges have one incident triangle. Eliminating $F$ gives the
claimed edge count. For $k=3$, subtracting the
three outer and three inner edges leaves six boundary--interior, or
\emph{cross}, edges. Each interior vertex already has two inner
neighbors and has total degree at least four, so each has exactly two
boundary neighbors. Each boundary vertex also has exactly two inner
neighbors. Indeed, zero would force the outer triangle to be a face.
If, for example, $A$ had only the inner neighbor $u$, its incident
triangles would be $ABu,ACu$, making $u$ adjacent to all three boundary
vertices, contrary to its cross degree two. All three boundary cross
degrees are therefore at least two, and their sum is six. The missing
cross edges form a perfect matching, which specifies the octahedral
disk.

Consider next a four-interior disk without an interior degree-three
vertex. Let $r$ be the number of inner--inner edges. There are fifteen
edges in all, hence $12-r$ cross edges. Counting contributions to the
sum of the four interior degrees gives
\[
                  2r+(12-r)=12+r\geq16,
\]
so $r\geq4$; simplicity gives $r\leq6$. A simple interior cycle
encloses a region that must be triangulated using interior labels:
boundary vertices lie outside it, and no edge can cross the cycle.

If $r=4$, the interior graph is connected, since a disconnected simple
graph on four vertices has at most three edges. A connected graph on
four vertices with four edges has one cycle. It is therefore either a
chordless four-cycle or a triangle with a pendant edge. The four-cycle
encloses no other label and requires a filling diagonal, contradicting
$r=4$.

In the second case, let $x$ be the pendant vertex. Its single inner
neighbor and total degree at least four force all three cross edges
$xA,xB,xC$. These split the outer triangle into the three sectors
$xAB,xBC,xCA$. The triangle on the other three interior labels is
connected and disjoint from these fan edges, so it lies wholly inside
one sector, say $xAB$. Every edge incident to one of those labels stays
in the closed sector: an edge to a vertex outside it would cross its
boundary. Their degrees in the sector disk thus equal their original
degrees and remain at least four. Apply the preceding three-interior
degree-and-cycle argument to this triangular disk, now with boundary
$x,A,B$. That argument used the inner triangle and degree bounds to
force the octahedral cross pattern. In particular, its boundary vertex
$x$ must have two neighbors among the three sector-interior labels.
This contradicts the single inner neighbor of $x$ in the original
disk, and excludes $r=4$.

If $r=6$, the interior graph is $K_4$. Its noncrossing straight
realization has one vertex inside the triangle of the other three;
four vertices in convex position would give crossing diagonals.
The enclosed vertex cannot have a boundary neighbor, since that edge
would cross the enclosing triangle. Its total degree is therefore
three, contrary to the hypothesis.

We are left with $r=5$, so the interior graph is $K_4$ minus one edge.
Its two three-cycles share an edge. They bound triangular faces on
opposite sides of that edge. To justify the side assertion, suppose
the remaining vertices were on the same side. Noncrossing then forces
one of the two triangles to be nested inside the other: otherwise
their other edges cross. The vertex enclosed by the larger triangle
would have to join all three of its corners to triangulate that region.
One of those joins is the missing edge, a contradiction. Thus the two
faces form a quadrilateral. Write its vertices cyclically as $x,y,z,w$,
with diagonal $yw$.

Every interior vertex needs an outer neighbor. More precisely, $x,z$
have inner degree two and need at least two outer neighbors, while
$y,w$ have inner degree three and need at least one. The seven cross
edges exceed these minima $(2,1,2,1)$ by exactly one. Symmetries of the
quadrilateral preserving its diagonal exchange the two nondiagonal
vertices and exchange the two diagonal vertices. Up to these
symmetries, the cross degrees in cyclic order are therefore
\[
                  (2,2,2,1)\quad\hbox{or}\quad(3,1,2,1).
\]

These lists determine the annular triangulations, not merely their
degree sequences. The region between the inner quadrilateral and outer
triangle has no further vertices or same-boundary diagonals. Each of
its triangles consequently has two cross edges and one boundary edge,
belonging to either the inner or outer boundary. Cut the annulus along
a cross edge and walk from one copy of that edge to the other through
successive triangles. The triangles form one cyclic strip before the
cut: across either cross edge there is a next triangle, and the
annulus is connected. A triangle on an inner boundary edge advances
the inner endpoint of the current cross edge; a triangle on an outer
edge advances its outer endpoint. The walk thus makes four inner
advances and three outer advances, each in the corresponding boundary's
fixed cyclic order.

At an inner vertex with $s$ outer neighbors, its cross edges occur
consecutively in its annular fan. Between arriving at that vertex and
advancing to the next inner vertex, the walk makes exactly $s-1$
outer advances. For the first degree list there is one outer advance
at each of $x,y,z$ and none at $w$. For the second there are two at
$x$, none at $y,w$, and one at $z$. Hence the degree list fixes the
circular advance word. Once its initial cross edge and the outer
boundary labels are chosen, every successive triangle is determined
by the endpoint that advances; changing those choices only relabels
the boundaries.

For clarity, the first walk, beginning at $Ax$, gives
\[
 ABx,\ Bxy,\ BCy,\ Cyz,\ CAz,\ Azw,\ Awx.
\]
The second, beginning at $Bx$, gives
\[
 BCx,\ CAx,\ Axy,\ Ayz,\ ABz,\ Bzw,\ Bwx.
\]
Adding the two inner faces $xyw,yzw$ gives exactly the two last rows
of Table~\ref{tab:diskforms}. In the first row, adding the outer face
$ABC$ gives a pentagonal bipyramid with poles $A,y$ and equatorial
cycle $B,x,w,z,C$. In the second, the triangles $xBC,xCA$ lie outside
$xAB$; the seven remaining triangles form an octahedral disk in that
sector, with inner labels $y,z,w$ and missing cross edges $xz,Aw,By$.

Finally, if the original four-interior disk has a degree-three vertex,
remove it once. The remainder is nonregular by Lemma~\ref{lem:ear}
and has three interior vertices, hence is octahedral by the first
part. Reinsertion subdivides the single face that replaced the removed
vertex. The octahedral disk has three caps, three annular triangles,
and one central triangle. Permuting its three missing
boundary--interior pairs carries any cap to any other cap and any
annular triangle to any other annular triangle. These are its three
face orbits. Subdivision therefore gives respectively the pendant,
annular-stellar, or central-stellar form in the first three table rows.
Together with the two degree-four-or-higher cases, this exhausts the
five forms.
\end{proof}

We give explicit data so that subsequent compatibility statements can
be checked without relying on a picture. With outer labels $A,B,C$ let
\[
 \mathcal O=\{ABc,BCa,CAb,Abc,Bac,Cab,abc\}.
\]
In the first three rows of Table~\ref{tab:diskforms}, the fourth interior
label is $d$. The \emph{cap label} on a boundary edge is the third vertex
of its unique incident triangle. The number $f$ counts triangles whose
three vertices are interior to the disk.

\begin{table}[htbp]
\centering\small
\begin{tabular}{>{\raggedright\arraybackslash}p{.20\linewidth}p{.59\linewidth}c}
\hline
Form & Triangles, or replacement in $\mathcal O$ & $f$\\
\hline
Pendant & Replace $ABc$ by $ABd,Acd,Bcd$ & 1\\
Annular-stellar & Replace $Abc$ by $Abd,Acd,bcd$ & 2\\
Central-stellar & Replace $abc$ by $abd,bcd,cad$ & 3\\
Pentagonal-bipyramid & $ABx,Bxy,BCy,Cyz,CAz,Azw,Awx,xyw,yzw$ & 2\\
Sector & $BCx,CAx,Axy,Ayz,ABz,Bzw,Bwx,xyw,yzw$ & 2\\
\hline
\end{tabular}
\caption{The five potentially nonregular four-interior triangular disks.
The triangle lists specify abstract incidence; any application uses the
actual noncrossing radial realization.}\label{tab:diskforms}
\end{table}

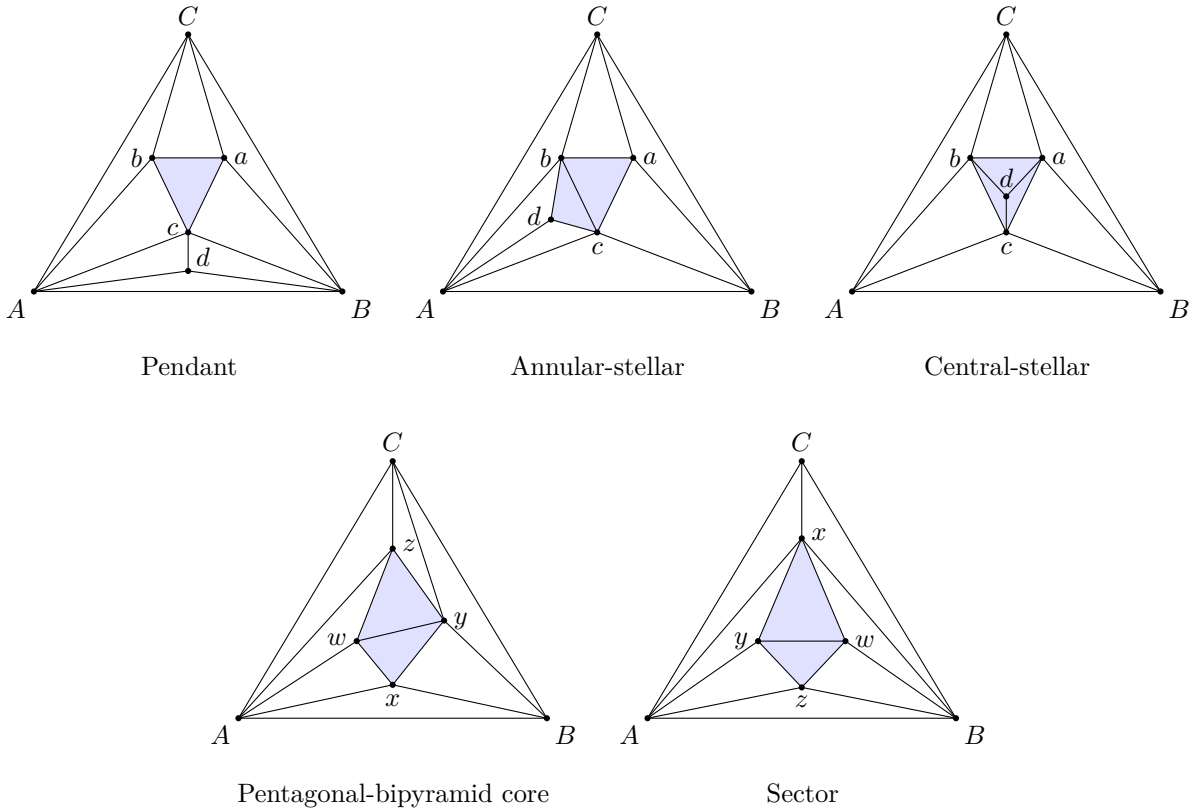
\begin{figure}[htbp]\centering
\begin{minipage}[b]{.32\textwidth}\centering
\begin{tikzpicture}[scale=0.68,every node/.style={font=\small},vertex label/.style={inner sep=1pt,outer sep=0pt}]
\coordinate (A) at (0.00000,0.00000);
\coordinate (B) at (6.00000,0.00000);
\coordinate (C) at (3.00000,5.00000);
\coordinate (a) at (3.70000,2.60000);
\coordinate (b) at (2.30000,2.60000);
\coordinate (c) at (3.00000,1.15000);
\coordinate (d) at (3.00000,0.40000);
\fill[blue!12] (a)--(b)--(c)--cycle;
\draw (A)--(B);
\draw (A)--(C);
\draw (A)--(b);
\draw (A)--(c);
\draw (A)--(d);
\draw (B)--(C);
\draw (B)--(a);
\draw (B)--(c);
\draw (B)--(d);
\draw (C)--(a);
\draw (C)--(b);
\draw (a)--(b);
\draw (a)--(c);
\draw (b)--(c);
\draw (c)--(d);
\fill (A) circle (1.7pt);
\node[vertex label,anchor=north east] at ([xshift=-3pt,yshift=-3pt]A) {$A$};
\fill (B) circle (1.7pt);
\node[vertex label,anchor=north west] at ([xshift=3pt,yshift=-3pt]B) {$B$};
\fill (C) circle (1.7pt);
\node[vertex label,anchor=south] at ([xshift=0pt,yshift=4pt]C) {$C$};
\fill (a) circle (1.7pt);
\node[vertex label,anchor=west] at ([xshift=4pt,yshift=0pt]a) {$a$};
\fill (b) circle (1.7pt);
\node[vertex label,anchor=east] at ([xshift=-4pt,yshift=0pt]b) {$b$};
\fill (c) circle (1.7pt);
\node[vertex label,anchor=east] at ([xshift=-4pt,yshift=2pt]c) {$c$};
\fill (d) circle (1.7pt);
\node[vertex label,anchor=south west] at ([xshift=3pt,yshift=1pt]d) {$d$};
\end{tikzpicture}
\par\vspace{3mm}\small Pendant
\end{minipage}
\begin{minipage}[b]{.32\textwidth}\centering
\begin{tikzpicture}[scale=0.68,every node/.style={font=\small},vertex label/.style={inner sep=1pt,outer sep=0pt}]
\coordinate (A) at (0.00000,0.00000);
\coordinate (B) at (6.00000,0.00000);
\coordinate (C) at (3.00000,5.00000);
\coordinate (a) at (3.70000,2.60000);
\coordinate (b) at (2.30000,2.60000);
\coordinate (c) at (3.00000,1.15000);
\coordinate (d) at (2.10000,1.40000);
\fill[blue!12] (a)--(b)--(c)--cycle;
\fill[blue!12] (b)--(c)--(d)--cycle;
\draw (A)--(B);
\draw (A)--(C);
\draw (A)--(b);
\draw (A)--(c);
\draw (A)--(d);
\draw (B)--(C);
\draw (B)--(a);
\draw (B)--(c);
\draw (C)--(a);
\draw (C)--(b);
\draw (a)--(b);
\draw (a)--(c);
\draw (b)--(c);
\draw (b)--(d);
\draw (c)--(d);
\fill (A) circle (1.7pt);
\node[vertex label,anchor=north east] at ([xshift=-3pt,yshift=-3pt]A) {$A$};
\fill (B) circle (1.7pt);
\node[vertex label,anchor=north west] at ([xshift=3pt,yshift=-3pt]B) {$B$};
\fill (C) circle (1.7pt);
\node[vertex label,anchor=south] at ([xshift=0pt,yshift=4pt]C) {$C$};
\fill (a) circle (1.7pt);
\node[vertex label,anchor=west] at ([xshift=4pt,yshift=0pt]a) {$a$};
\fill (b) circle (1.7pt);
\node[vertex label,anchor=east] at ([xshift=-4pt,yshift=0pt]b) {$b$};
\fill (c) circle (1.7pt);
\node[vertex label,anchor=north] at ([xshift=0pt,yshift=-4pt]c) {$c$};
\fill (d) circle (1.7pt);
\node[vertex label,anchor=east] at ([xshift=-4pt,yshift=2pt]d) {$d$};
\end{tikzpicture}
\par\vspace{3mm}\small Annular-stellar
\end{minipage}
\begin{minipage}[b]{.32\textwidth}\centering
\begin{tikzpicture}[scale=0.68,every node/.style={font=\small},vertex label/.style={inner sep=1pt,outer sep=0pt}]
\coordinate (A) at (0.00000,0.00000);
\coordinate (B) at (6.00000,0.00000);
\coordinate (C) at (3.00000,5.00000);
\coordinate (a) at (3.70000,2.60000);
\coordinate (b) at (2.30000,2.60000);
\coordinate (c) at (3.00000,1.15000);
\coordinate (d) at (3.00000,1.85000);
\fill[blue!12] (a)--(b)--(d)--cycle;
\fill[blue!12] (b)--(c)--(d)--cycle;
\fill[blue!12] (c)--(a)--(d)--cycle;
\draw (A)--(B);
\draw (A)--(C);
\draw (A)--(b);
\draw (A)--(c);
\draw (B)--(C);
\draw (B)--(a);
\draw (B)--(c);
\draw (C)--(a);
\draw (C)--(b);
\draw (a)--(b);
\draw (a)--(c);
\draw (a)--(d);
\draw (b)--(c);
\draw (b)--(d);
\draw (c)--(d);
\fill (A) circle (1.7pt);
\node[vertex label,anchor=north east] at ([xshift=-3pt,yshift=-3pt]A) {$A$};
\fill (B) circle (1.7pt);
\node[vertex label,anchor=north west] at ([xshift=3pt,yshift=-3pt]B) {$B$};
\fill (C) circle (1.7pt);
\node[vertex label,anchor=south] at ([xshift=0pt,yshift=4pt]C) {$C$};
\fill (a) circle (1.7pt);
\node[vertex label,anchor=west] at ([xshift=4pt,yshift=0pt]a) {$a$};
\fill (b) circle (1.7pt);
\node[vertex label,anchor=east] at ([xshift=-4pt,yshift=0pt]b) {$b$};
\fill (c) circle (1.7pt);
\node[vertex label,anchor=north] at ([xshift=0pt,yshift=-4pt]c) {$c$};
\fill (d) circle (1.7pt);
\node[vertex label,anchor=south] at ([xshift=0pt,yshift=4pt]d) {$d$};
\end{tikzpicture}
\par\vspace{3mm}\small Central-stellar
\end{minipage}
\par\vspace{7mm}
\begin{minipage}[b]{.32\textwidth}\centering
\begin{tikzpicture}[scale=0.68,every node/.style={font=\small},vertex label/.style={inner sep=1pt,outer sep=0pt}]
\coordinate (A) at (0.00000,0.00000);
\coordinate (B) at (6.00000,0.00000);
\coordinate (C) at (3.00000,5.00000);
\coordinate (w) at (2.30000,1.50000);
\coordinate (x) at (3.00000,0.65000);
\coordinate (y) at (4.00000,1.90000);
\coordinate (z) at (3.00000,3.30000);
\fill[blue!12] (x)--(y)--(w)--cycle;
\fill[blue!12] (y)--(z)--(w)--cycle;
\draw (A)--(B);
\draw (A)--(C);
\draw (A)--(w);
\draw (A)--(x);
\draw (A)--(z);
\draw (B)--(C);
\draw (B)--(x);
\draw (B)--(y);
\draw (C)--(y);
\draw (C)--(z);
\draw (w)--(x);
\draw (w)--(y);
\draw (w)--(z);
\draw (x)--(y);
\draw (y)--(z);
\fill (A) circle (1.7pt);
\node[vertex label,anchor=north east] at ([xshift=-3pt,yshift=-3pt]A) {$A$};
\fill (B) circle (1.7pt);
\node[vertex label,anchor=north west] at ([xshift=3pt,yshift=-3pt]B) {$B$};
\fill (C) circle (1.7pt);
\node[vertex label,anchor=south] at ([xshift=0pt,yshift=4pt]C) {$C$};
\fill (w) circle (1.7pt);
\node[vertex label,anchor=east] at ([xshift=-4pt,yshift=1pt]w) {$w$};
\fill (x) circle (1.7pt);
\node[vertex label,anchor=north] at ([xshift=0pt,yshift=-4pt]x) {$x$};
\fill (y) circle (1.7pt);
\node[vertex label,anchor=west] at ([xshift=4pt,yshift=0pt]y) {$y$};
\fill (z) circle (1.7pt);
\node[vertex label,anchor=west] at ([xshift=4pt,yshift=2pt]z) {$z$};
\end{tikzpicture}
\par\vspace{3mm}\small Pentagonal-bipyramid core
\end{minipage}
\begin{minipage}[b]{.32\textwidth}\centering
\begin{tikzpicture}[scale=0.68,every node/.style={font=\small},vertex label/.style={inner sep=1pt,outer sep=0pt}]
\coordinate (A) at (0.00000,0.00000);
\coordinate (B) at (6.00000,0.00000);
\coordinate (C) at (3.00000,5.00000);
\coordinate (w) at (3.85000,1.50000);
\coordinate (x) at (3.00000,3.50000);
\coordinate (y) at (2.15000,1.50000);
\coordinate (z) at (3.00000,0.60000);
\fill[blue!12] (x)--(y)--(w)--cycle;
\fill[blue!12] (y)--(z)--(w)--cycle;
\draw (A)--(B);
\draw (A)--(C);
\draw (A)--(x);
\draw (A)--(y);
\draw (A)--(z);
\draw (B)--(C);
\draw (B)--(w);
\draw (B)--(x);
\draw (B)--(z);
\draw (C)--(x);
\draw (w)--(x);
\draw (w)--(y);
\draw (w)--(z);
\draw (x)--(y);
\draw (y)--(z);
\fill (A) circle (1.7pt);
\node[vertex label,anchor=north east] at ([xshift=-3pt,yshift=-3pt]A) {$A$};
\fill (B) circle (1.7pt);
\node[vertex label,anchor=north west] at ([xshift=3pt,yshift=-3pt]B) {$B$};
\fill (C) circle (1.7pt);
\node[vertex label,anchor=south] at ([xshift=0pt,yshift=4pt]C) {$C$};
\fill (w) circle (1.7pt);
\node[vertex label,anchor=west] at ([xshift=4pt,yshift=0pt]w) {$w$};
\fill (x) circle (1.7pt);
\node[vertex label,anchor=west] at ([xshift=4pt,yshift=2pt]x) {$x$};
\fill (y) circle (1.7pt);
\node[vertex label,anchor=east] at ([xshift=-4pt,yshift=0pt]y) {$y$};
\fill (z) circle (1.7pt);
\node[vertex label,anchor=north] at ([xshift=0pt,yshift=-3pt]z) {$z$};
\end{tikzpicture}
\par\vspace{3mm}\small Sector
\end{minipage}
\caption{The five potentially nonregular four-interior disk forms, schematically. All-interior faces are shaded. Straight-line layouts illustrate incidence only; strict regularity is always decided in the actual projected coordinates.}\label{fig:forms}
\end{figure}

Several consequences of these lists will be used. A pendant disk has
an interior triangle with one pendant edge, and its pendant vertex has
total planar degree three. Its three cap labels are distinct, and its
pendant vertex is one of them. The unique interior neighbor of the
pendant is the interior label absent from the current cap triple.
Every $f=2$ form has seven cross edges and every inner vertex has an
outer neighbor. Only the sector has a repeated cap: on $BC,CA,AB$ its
caps are $x,x,z$. Its all-interior faces are $xyw,yzw$; the boundary
vertices $A,B,C$ have respectively $3,3,1$ inner neighbors. The other
$f=2$ forms have distinct caps and boundary cross-degree multiset
$\{3,2,2\}$. The central-stellar form has $f=3$, and no form has
$f>3$.

The auxiliary degree-three equivalence applies to the three stellar
forms. Similarly, the sector is regular exactly when its octahedral
subdisk is regular: restriction proves necessity, while for sufficiency
one first strictly lifts the outer stellar disk at $x$, then inserts a
sufficiently small affine-normalized strict lifting inside $xAB$.
Strict boundary folds persist by openness. No universal regularity
claim about the pentagonal-bipyramid core is needed.

\Needspace{12\baselineskip}
\begin{lemma}[Tetrahedral-hull incidence]\label{lem:hullincidence}
Let $H=\{H_0,H_1,H_2,H_3\}$ be the hull and let $I$ consist of four
interior labels. Write $a_{HI}$ and $b_{II}$ for the numbers of actual
hull--interior and interior--interior edges. Let $t_j$ count tetrahedra
with $j$ interior vertices, and put $\epsilon=t_4\in\{0,1\}$.
Then
\begin{align}
 t&=a_{HI}+b_{II}-3,& t_1&=4,\label{eq:incidence-a}\\
 t_2&=a_{HI}-b_{II}-1+\epsilon,&
 t_3&=2b_{II}-6-2\epsilon.\label{eq:incidence-b}
\end{align}
These are necessary incidence identities, not sufficient conditions
for a geometric realization.
\end{lemma}
\begin{proof}
The hull has six edges and four faces. For the triangulated ball,
$8-e+f-t=1$ and $4t=2(f-4)+4$, giving
$t=e-9=a_{HI}+b_{II}-3$. A tetrahedron with no interior label would
occupy the whole hull and contain active labels, so $t_0=0$.
Every tetrahedron with three hull labels meets a hull facet, and each
of the four facets has one tetrahedral coface. Thus $t_1=4$.
If $k_q$ is the number of interior neighbors of a hull vertex $q$, its
triangular link has $2k_q+1$ triangles. Counting incidences with hull
vertices yields
\[
          3t_1+2t_2+t_3=\sum_{q\in H}(2k_q+1)=2a_{HI}+4.
\]
Together with $t=t_1+t_2+t_3+\epsilon$, this gives
\eqref{eq:incidence-b}. Counting triangles in the spherical interior
links gives the consistent identity
\[
 t_1+2t_2+3t_3+4\epsilon=2(a_{HI}+2b_{II})-16.
\]
\end{proof}

\begin{theorem}[Compatibility of four hull links]\label{thm:compatibility}
Let $T$ be a full tetrahedralization of eight points with tetrahedral
hull and four interior labels. If all sixteen hull--interior edges
occur, at least one of its four actual hull radial links is strictly regular.
\end{theorem}
\begin{proof}
Suppose all four are nonregular. Let $f_i$ count all-interior triangles
in the radial link of $H_i$. Each such triangle gives one tetrahedron
with exactly one hull label, so
\begin{equation}\label{eq:fcount}
                  \sum_{i=0}^3 f_i=t_3=2b_{II}-6-2\epsilon.
\end{equation}
By Lemma~\ref{lem:cores}, $1\leq f_i\leq3$. If $\epsilon=1$, then
$b_{II}=6$ and all $f_i=1$. If $\epsilon=0$, then $b_{II}\geq5$;
$b_{II}=5$ again makes every $f_i=1$. For $b_{II}=6$, the only
patterns are $(3,1,1,1)$ and $(2,2,1,1)$, up to permutation.
We now exclude these possibilities by comparing actual spatial faces.

For each $i$, let $c_i\in I$ be the apex of the unique tetrahedron on
the hull facet opposite $H_i$. For distinct $i,j$, write $H_k,H_l$ for
the other two hull labels. The triangle $H_kH_lc_j$ in the $H_i$ link
cones to the actual tetrahedron $H_iH_kH_lc_j$, whose hull facet is
opposite $H_j$. Hence the cap opposite the boundary vertex $H_j$ in
that link is $c_j$.

The same correspondence applies to edges of a link: $x$ is an interior
neighbor of its boundary vertex $H_j$ at pivot $H_i$ if and only if
$H_iH_jx$ is a triangular face of $T$. At pivot $H_j$, the test is again
membership of that very same face $H_iH_jx$. Thus the two links give
identical subsets of the four original interior labels, regardless of
the names used to display their disk forms.

\emph{Three pendant links cannot coexist.}
If three pivots have pendant links, all four labels $c_i$ are distinct.
Indeed, each pendant cap triple is distinct, and every pair $c_i,c_j$
belongs to the cap triple of at least one of the three pivots.
For a pendant pivot $H_i$, its cap triple is $\{c_j:j\ne i\}$, so the
interior label omitted from that triple is $c_i$. It is the unique
interior neighbor of the pendant vertex. Write the pendant as
$c_{s(i)}$, where $s(i)\ne i$.

At the boundary vertex $H_{s(i)}$, the old octahedral disk lacked its
opposite cap $c_i$; the new pendant is also absent from that vertex.
At any other boundary vertex $H_j$, the old missing neighbor $c_j$
remains missing, and the pendant is present. Consequently the absent
spatial faces $H_iH_jx$, for $x\in I$, have exactly the label sets
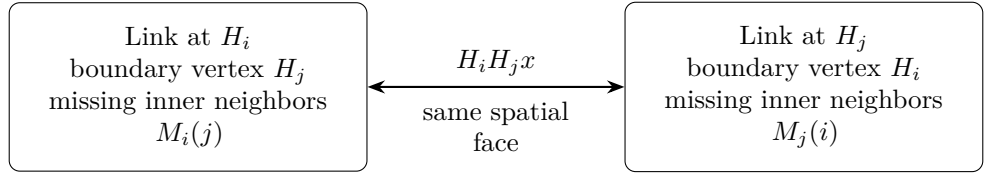
\begin{equation}\label{eq:pendantmissing}
 \{x\in I:H_iH_jx\notin T\}=
 \begin{cases}
   \{c_i,c_j\},&j=s(i),\\
   \{c_j\},&j\ne s(i).
 \end{cases}
\end{equation}
\begin{figure}[htbp]\centering
\begin{tikzpicture}[>=Stealth,font=\small]
\node[draw,rounded corners,align=center,text width=4.1cm,inner sep=9pt] (left)
{Link at $H_i$\\boundary vertex $H_j$\\missing inner neighbors\\$M_i(j)$};
\node[draw,rounded corners,align=center,text width=4.1cm,inner sep=9pt,right=3.4cm of left] (right)
{Link at $H_j$\\boundary vertex $H_i$\\missing inner neighbors\\$M_j(i)$};
\draw[<->,thick] (left)--node[above=3pt,inner sep=1pt] {$H_iH_jx$} node[below=4pt,align=center,text width=2.8cm,inner sep=1pt] {same spatial\\face} (right);
\node[below=9mm of left,align=center] {$M_i(j)=\{c_j\}$ if $j\ne s(i)$\\[3pt]
$M_i(s(i))=\{c_i,c_{s(i)}\}$};
\node[below=9mm of right,align=center] {$M_j(i)=\{c_i\}$ if $i\ne s(j)$\\[3pt]
$M_j(s(j))=\{c_j,c_{s(j)}\}$};
\end{tikzpicture}
\caption{Pendant cap compatibility for distinct $c_0,c_1,c_2,c_3$.
Here $c_k$ is the actual apex on the hull facet opposite $H_k$, and
$c_{s(i)}$ is the pendant at pivot $H_i$. Equality of the missing-face
sets forces $s(i)=j$ and $s(j)=i$. Three pendant pivots cannot pair in
this way. This schematic compares actual shared simplices.}
\label{fig:caps}
\end{figure}

If $H_i,H_j$ are both pendant pivots, their descriptions concern the
same four possible triangular faces $H_iH_jx$. With neither arrow
$s(i)=j$, $s(j)=i$, the two singleton sets differ; with just one arrow
the set sizes differ. Thus both arrows must occur. Three pendant
pivots cannot be pairwise paired in this way. This rules out all
patterns with at least three entries equal to one.

\emph{The residual pattern is $(2,2,1,1)$.}
Let $H_0,H_1$ be the pendant pivots. Distinctness of their two cap
triples implies that $c_2,c_3$ differ from each other and from both
$c_0,c_1$. Thus either all four caps are distinct or $c_0=c_1$.

If all are distinct, \eqref{eq:pendantmissing} forces
$s(0)=1$ and $s(1)=0$. Along either spatial hull edge $H_0H_2$ or
$H_1H_2$, the present triangular faces are exactly those indexed by
$I\setminus\{c_2\}$. In the $H_2$ disk, both boundary vertices
$H_0,H_1$ therefore have three inner neighbors and miss $c_2$.
An $f=2$ disk has seven cross edges, and each inner label has an outer
neighbor. The two rows at $H_0,H_1$ already account for six of those
edges, leaving exactly one at $H_3$. Both rows omit $c_2$, so this
remaining neighbor must be $c_2$; otherwise $c_2$ would have no outer
neighbor at all. Thus the unique triangular
face $H_2H_3x$ is $H_2H_3c_2$. The identical argument at $H_3$ makes
it $H_2H_3c_3$. These are incompatible descriptions of the same
spatial faces since $c_2\ne c_3$.

It remains to consider $c_0=c_1=v$. Put $c_2=w,c_3=x$, and denote the
fourth interior label by $y$. The two $f=2$ disks have repeated cap
labels, so both are sector forms. At $H_2$ the caps indexed $0,1,3$
are $v,v,x$, and its all-interior faces are $vwy,wxy$. At $H_3$ the
caps are $v,v,w$, and its all-interior faces are $vxy,wxy$.
To read these from the sector row, its repeated cap is one vertex of
the first interior triangle, its single cap is one vertex of the second,
and the two labels absent from its caps form their common edge.
Interchanging those last two labels does not change the pair of faces.
Thus the actual triangle $wxy$ already has the two tetrahedral cofaces
$H_2wxy,H_3wxy$. Neither pendant pivot can have pendant label $v$,
because its only all-interior triangle would then be $wxy$, supplying
a third coface. Its pendant is a current cap, hence is not $y$ either.
Each pendant label is therefore $w$ or $x$.

If the $H_0$ pendant is $w$, it was inserted into the cap opposite
boundary vertex $H_2$. Its unique inner neighbor is the missing cap
label $y$. At that boundary vertex both $w$ and $y$ are absent, so
exactly two faces $H_0H_2z$, $z\in I$, occur. If its pendant is $x$,
exactly two faces $H_0H_3z$ occur instead. But in both sector disks,
the boundary vertices $H_0,H_1$ each have three inner neighbors.
Indeed, at pivot $H_2$ the repeated-cap edges are $H_1H_3,H_0H_3$,
so their meeting vertex $H_3$ is the one-neighbor boundary vertex of
the sector. At pivot $H_3$ the repeated-cap edges are
$H_1H_2,H_0H_2$, and that exceptional vertex is $H_2$ instead.
In both cases the rows at $H_0,H_1$ therefore have size three.
The disk at $H_2$ or $H_3$ gives three faces for
the same spatial hull edge that the pendant disk gives two. This last
contradiction excludes the repeated-cap case and completes the proof.
\end{proof}

\section{Completion of the eight-point theorem}
\label{sec:completion}
\begin{corollary}[Eight-point placing bound]\label{cor:35}
Every full tetrahedralization of eight points in the stated general
position reaches placing form at some hull vertex in at most $35$ legal
flips.
\end{corollary}
\begin{proof}
Write $h$ for the number of hull vertices. For $h\geq5$,
Theorem~\ref{thm:hfiveplus} supplies an actual strictly regular radial
link. Its pivot has at most seven mesh neighbors, so
Theorem~\ref{thm:radial} gives at most $\binom74=35$ flips.
If $h=4$ and some hull vertex has at most three interior neighbors,
Theorem~\ref{thm:octprep} gives at most $1+\binom64=16$ flips (at most
$15$ when its initial link is regular). Otherwise all sixteen
hull/interior edges occur. Theorem~\ref{thm:compatibility} supplies a
strictly regular link with seven neighbors, again giving $35$.
Every possible hull distribution appears in Table~\ref{tab:hulls}.
\end{proof}

\begin{table}[htbp]\centering\small
\begin{tabular}{@{}ccp{7.3cm}r@{}}
\toprule
Hull & Interior & Mechanism at eight points & To placing\\
\midrule
8 & 0 & Hull degree at least five; forest link & $\leq35$\\
7 & 1 & Hull degree at least five; forest link & $\leq35$\\
6 & 2 & Forced octahedral hull under all-nonregularity;
auxiliary universal-ear contradiction & $\leq35$\\
5 & 3 & Equatorial shared-triangle contradiction & $\leq35$\\
4 & 4 & A pivot misses an interior neighbor: positive-target
preparation and a six-neighbor sweep & $\leq16$\\
4 & 4 & Complete hull/interior adjacency: four-link compatibility
and a seven-neighbor sweep & $\leq35$\\
\bottomrule
\end{tabular}
\caption{All eight-point hull distributions. The first four rows use
Theorems~\ref{thm:forest}, \ref{thm:hfiveplus}, and \ref{thm:radial};
the last two use Theorems~\ref{thm:octprep} and
\ref{thm:compatibility}, respectively, followed by \ref{thm:radial}.
The two last rows partition the tetrahedral-hull case.}
\label{tab:hulls}
\end{table}

\begin{proof}[Proof of Theorem~\ref{thm:main}]
Theorem~\ref{thm:seven} supplies full connectivity for every configuration
of at most seven points by radial placing induction, starting with the
single tetrahedron on four points. Let $|P|=8$ and choose any full
triangulation $T$. Corollary~\ref{cor:35} connects $T$ to a placing
triangulation at some hull vertex $q$. Its antistar is a full
triangulation of the seven-point convex hull
$\conv(P\setminus\{q\})$. By Theorem~\ref{thm:seven} and
Lemma~\ref{lem:placing}, $T$ is connected to a full regular triangulation
of $P$. Lemma~\ref{lem:regular} connects all these regular endpoints,
even when different starts chose different pivots. All steps are full
geometric $2\leftrightarrow3$ flips on the original coordinates.
\end{proof}
The additional antistar path is not counted by Corollary~\ref{cor:35}.
The weak unit-pivot height at placing is useful for identifying its
convex geometry; the explicit regular extension, rather than strictness
of that weak height, completes the connectivity proof.

\section*{Acknowledgements}
This paper was produced in collaboration with GPT 6 Astra using ultra reasoning, and all generated work was checked by the author.
\begingroup\small
\bibliographystyle{plain}
\bibliography{references}

\begin{thebibliography}{1}

\bibitem{AS2000}
Miguel Azaola and Francisco Santos.
\newblock The graph of triangulations of a point configuration with {$d+4$}
  vertices is 3-connected.
\newblock {\em Discrete \& Computational Geometry}, 23(4):489--536, 2000.

\bibitem{DRS2010}
Jes{\'u}s~A. De~Loera, J{\"o}rg Rambau, and Francisco Santos.
\newblock {\em Triangulations: Structures for Algorithms and Applications},
  volume~25 of {\em Algorithms and Computation in Mathematics}.
\newblock Springer, Berlin, Heidelberg, 2010.

\bibitem{DSU1999}
Jes{\'u}s~A. De~Loera, Francisco Santos, and Jorge Urrutia.
\newblock The number of geometric bistellar neighbors of a triangulation.
\newblock {\em Discrete \& Computational Geometry}, 21(1):131--142, 1999.

\bibitem{PL2007}
Lionel Pournin and Thomas~M. Liebling.
\newblock Constrained paths in the flip-graph of regular triangulations.
\newblock {\em Computational Geometry}, 37(2):134--140, 2007.

\bibitem{Santos2006}
Francisco Santos.
\newblock Geometric bistellar flips: The setting, the context and a
  construction.
\newblock In Marta Sanz-Sol{\'e}, Javier Soria, Juan~Luis Varona, and Joan
  Verdera, editors, {\em Proceedings of the International Congress of
  Mathematicians 2006}, volume III, pages 931--962. European Mathematical
  Society, 2006.

\end{thebibliography}
\endgroup
\end{document}